\documentclass[11pt,a4paper,reqno]{amsart}%
\usepackage{amsthm,amsmath,amsfonts,amssymb,amsxtra,appendix,bookmark,dsfont,latexsym,color}

\makeatletter
\usepackage{hyperref}
\usepackage{delarray,a4,color}
\usepackage{euscript}
\usepackage[latin1]{inputenc}

\usepackage{pdfsync}
\numberwithin{equation}{section}

\newtheorem{theorem}{Theorem}[section]
\newtheorem{lemma}[theorem]{Lemma}
\newtheorem{corollary}[theorem]{Corollary}
\newtheorem{proposition}[theorem]{Proposition}

\theoremstyle{definition}

\theoremstyle{remark}
\newtheorem{remark}[theorem]{Remark}

\numberwithin{equation}{section}

\usepackage{color}

\newcommand{\red}{\color{red}}

\numberwithin{equation}{section}

\newcommand{\bdm}{\begin{displaymath}}
\newcommand{\edm}{\end{displaymath}}
\newcommand{\bdn}{\begin{eqnarray}}
\newcommand{\edn}{\end{eqnarray}}
\newcommand{\bay}{\begin{array}{c}}
\newcommand{\eay}{\end{array}}
\newcommand{\ben}{\begin{enumerate}}
\newcommand{\een}{\end{enumerate}}
\newcommand{\beq}{\begin{equation}}
\newcommand{\eeq}{\end{equation}}

\newcommand{\tx}{\textstyle}

\newcommand{\eps}{\varepsilon}

\newcommand{\chiin}{\chi_{\mathrm{in}}}
\newcommand{\chiout}{\chi_{\mathrm{out}}}

\newcommand{\R}{\mathbb{R}}
\newcommand{\N}{\mathbb{N}}

\newcommand{\C}{\mathbb{C}}

\newcommand{\F}{\mathcal{F}}
\newcommand{\E}{\mathcal{E}}
\newcommand{\A}{\mathcal{A}}
\newcommand{\B}{\mathcal{B}}

\newcommand{\I}{\mathcal{I}}

\newcommand{\V}{\mathcal{V}}

\newcommand{\PP}{\mathcal{P}}

\newcommand{\xbf}{\mathbf{x}}

\newcommand{\one}{{\ensuremath {\mathds 1} }}

\newcommand{\al}{\alpha}

\newcommand{\ep}{\varepsilon}

\newcommand{\Om}{\Omega}
\newcommand{\om}{\omega}

\newcommand{\half}{\frac{1}{2}}

\newcommand{\supp}{\mathrm{supp}}
\newcommand{\wto}{\rightharpoonup}

\newcommand{\intR}{\int_{\R ^2}}
\newcommand{\intRN}{\int_{\R ^{2N}}}

\newcommand{\rhoP}{\rho_{\Psi}}
\newcommand{\rhoPF}{\rho_{\Psi_F}}

\newcommand{\LLL}{\mathfrak{H}}

\newcommand{\Barg}{\B}
\newcommand{\BargN}{\B ^{N}}

\newcommand{\cLau}{c _{\rm Lau}}
\newcommand{\intN}{\I_N}

\newcommand{\Ker}{\mathrm{Ker} (\intN)}

\newcommand{\AN}{\A_{N}}

\newcommand{\rhoMFel}{\varrho ^{\rm el}}

\newcommand{\Rminus}{R_m ^-}
\newcommand{\Rplus}{R_m ^+}

\newcommand{\PsiLau}{\Psi_{\rm Lau}}

\newcommand{\PPs}{\mathcal{P}_{\rm sym}}

\newcommand{\rhoF}{\rho_F}

\newcommand{\muF}{\mu_F}
\newcommand{\muFk}{\mu_F ^{(k)}}
\newcommand{\muFone}{\mu_F ^{(1)}}
\newcommand{\ZF}{\mathcal{Z}_F}

\newcommand{\MFf}{\E ^{\rm MF}}
\newcommand{\MFe}{E ^{\rm MF}}
\newcommand{\rhoMF}{\varrho ^{\rm MF}}

\newcommand{\logal}{\log_{\al}}
\newcommand{\Wal}{W_{\al}}
\newcommand{\mut}{\tilde{\mu}}

\newcommand{\Hepal}{H_{\eps,\alpha}}
\newcommand{\Fepal}{\F_{\eps,\alpha}}
\newcommand{\muepal}{\mu_{\ep,\al}}
\newcommand{\Zepal}{\mathcal{Z}_{\ep,\al}}
\newcommand{\MFfepal}{\MFf_{\eps,\alpha}}
\newcommand{\MFeepal}{\MFe_{\eps,\alpha}}
\newcommand{\rhoMFepal}{\rhoMF _{\eps,\alpha}}

\newcommand{\mb}{\underline{m}}

\newcommand{\error}{\mathrm{Err}}
\newcommand{\VD}{\mathcal{V}_2 ^D}
\newcommand{\ED}{E_2 ^D (N)}

\numberwithin{equation}{section}

\title{Incompressibility estimates for the Laughlin phase}

\author[Nicolas Rougerie]{Nicolas Rougerie}
\address{Universit\'e Grenoble 1 \& CNRS, LPMMC, UMR 5493, BP 166, 38042 Grenoble, France.}

\author[Jakob Yngvason]{Jakob Yngvason}
\address{Fakult\"at f\"ur Physik, Universit{\"a}t Wien, Boltzmanngasse 5, 1090 Vienna, Austria \\
Erwin Schr{\"o}dinger Institute for Mathematical Physics, Boltzmanngasse 9, 1090 Vienna, Austria.}

\date{November 20, 2014}

\begin{document}

\begin{abstract}
This paper has its motivation in the study of the Fractional Quantum Hall Effect. We consider 2D quantum particles submitted to a strong perpendicular magnetic field, reducing admissible wave functions to those of the Lowest Landau Level. When repulsive interactions are strong enough in this model, highly correlated states emerge, built on Laughlin's famous wave function.  We investigate a model for the response of such strongly correlated ground states to variations of an external potential. This leads to a family of variational problems of a new type. Our main results are rigorous energy estimates demonstrating a strong rigidity of the response of strongly correlated states to the external potential. In particular we obtain estimates indicating that there is a universal bound on the maximum local density of these states in the limit of large particle number.  We refer to these as incompressibility estimates. 
\end{abstract}

\maketitle

\tableofcontents

\section{Introduction}

The introduction by Laughlin of his famous wave-function~\cite{Lau,Lau2} forms the starting point of our current theoretical understanding of the fractional quantum Hall effect (FQHE)~\cite{BF,Goe,Gir,STG}, one the most intriguing phenomena of condensed matter physics. The situation of interest is that of interacting particles confined in two space dimensions, submitted to a perpendicular magnetic field of constant strength. The Hamiltonian of the full system is given by
\begin{equation}\label{eq:first princ hamil}
H_N = \sum_{j=1} ^N \left( - \left( \nabla_j -  \mathrm i\mathbf A(\mathbf x_j)
\right) ^2 + V (\xbf_j) \right) + \lambda \sum_{1\leq i<j \leq N} w (\xbf_i-\xbf_j) 
\end{equation}
where $\mathbf A$ is the vector potential of the applied magnetic field of strength $B$, given by
$$\mathbf{A}(\mathbf {x})=  \tx\frac{B}{2} (-y,x)$$
in the symmetric gauge. The pair interaction potential is denoted $w$, the coupling constant~$\lambda$. Here we think of repulsive interactions, $w\geq 0, \lambda \geq 0$.  Units are chosen so that  Planck's constant $\hbar$, the velocity of light and the charge are equal to 1 and the mass is equal to 1/2. The potential $V$ can model both trapping of the particles and disorder in the sample. We will study systems made of fermions or bosons, so $H_N$ will act either on $\bigotimes_{\rm as}  ^N L ^2 (\R ^2)$ or $\bigotimes_{\rm s}  ^N L ^2 (\R ^2)$, the anti-symmetric or symmetric $N$-fold tensor product of~$L ^2 (\R ^2)$.  

The quantum Hall effect (integer or fractional) occurs in the regime of very large applied magnetic fields, $B\gg 1$. In this case one can restrict attention to single particle states corresponding to the lowest eigenvalue of $ - \left( \nabla - i \mathbf A(\mathbf x)\right) ^2$, i.e., the lowest Landau level~(LLL)
\begin{equation}\label{eq:intro LLL}
\LLL := \left\{ \psi\in L ^2 (\R ^2):\ \psi (\xbf) = f(z) e ^{ - B |z| ^2 / 4},\: f \mbox{ holomorphic } \right\} 
\end{equation}
as the single-particle Hilbert space. Here and in the sequel we identify vectors $\xbf\in \R ^2$ and complex numbers $z =x+iy \in \C$. Henceforth we choose  units so that the magnetic field is $2$, in order to comply with the conventions for rotating Bose gases as in~\cite{RSY1, RSY2} where $B$ corresponds to $2$ times the angular velocity.

While the integer quantum Hall effect (IQHE) can be understood in terms of single-particle physics and the Pauli principle\footnote{The IQHE hence requires the particles to be fermions.} only, the FQHE has its origin in the pair interactions between particles\footnote{The FQHE can thus in principle also occur in bosonic systems.}. The situation that is best understood is that where the interactions are repulsive enough to force the many-body  wave-function to vanish whenever particles come close together. This leads to the following trial states~\cite{Lau}
\begin{equation}\label{eq:intro Laughlin}
\PsiLau (z_1,\ldots,z_N) = \cLau \prod_{1\leq i<j \leq N} (z_i-z_j) ^{\ell}  e ^{- \sum_{j=1} ^N |z_j| ^2 / 2} 
\end{equation}
where $\ell$ is an odd number for fermions (the case originally considered by Laughlin is $\ell = 3$) and an even number for bosons ($\ell = 2$ is the most relevant), as imposed by the symmetry. The constant $\cLau$ normalizes the state in $L^2(\R ^{2N})$. Much of FQHE physics is based on the strong inter-particle correlations included in Laughlin's wave-function.

More generally, one may restrict the wave function to live in the space 
\begin{equation}\label{eq:Ker}
\Ker = \left\{  \Psi\in L^2(\R ^{2N}):\  \Psi (z_1,\ldots,z_N) = \PsiLau (z_1,\ldots,z_N) F (z_1,\ldots,z_N), \: F \in \BargN \right\} 
\end{equation}
where 
\begin{equation}\label{eq:BargN}
\BargN := \bigotimes_s  ^N \Barg =  \left\{ F \mbox{ holomorphic and symmetric } | \:  F(z_1,\ldots,z_N) e^{- \sum_{j=1} ^N |z_j| ^2 /2  } \in L ^2 (\R ^{2N})\right\} 
\end{equation}
is the $N$-body bosonic Bargmann space (symmetric here means invariant under exchange of two particles $z_i$, $z_j$), with scalar product 
\begin{equation}\label{eq:scalar BargN}
\left\langle F, G \right\rangle_{\BargN} : = \left\langle F e^{-\sum_{j=1} ^N |z_j| ^2 /2} , G e^{-\sum_{j=1} ^N |z_j| ^2 /2}  \right\rangle_{L^2 (\R ^{2N})}.
\end{equation}
The restriction of admissible states to $\Ker$ is reminiscent of the use of the Gutzwiller projector in the study of Mott insulators, see~\cite{LNW} and references therein.

The choice of the notation $\Ker$ is motivated by the fact that for $\ell = 2$ this is precisely the kernel of the interaction operator
$$\mathcal I_N = \sum_{1\leq i <j \leq N} w(\xbf_i - \xbf_j)$$ 
acting on the $N$-body bosonic LLL, with $w$ a contact potential $\delta(\xbf)$~\cite{TK,PB,RSY2}. For bosons with short range interactions, the contact interaction can be rigorously proved to emerge in a well defined limit~\cite{LS}. We remark that the wave function is an eigenfunction of angular momentum if and only if the correlation factor is a homogeneous polynomial.

The vanishing of $\PsiLau$ along the diagonals $z_i=z_j$ strongly decreases the interaction energy, and in this paper we will neglect all the eventual residual interaction, an approximation which is frequently made in the literature~\cite{Lau,Jai}. This leads us to the study of a very simple energy functional 
\begin{equation}\label{eq:start energy}
\E [\Psi] = N \intR V (z) \rhoP (z) dz  
\end{equation}
depending only on the 1-particle probability density $\rhoP$ 
\begin{equation}\label{eq:intro density}
\rhoP (z):= \int_{\R ^{2(N-1)}} |\Psi (z,z_2,\ldots,z_N)| ^2 dz_2 \ldots dz_N 
\end{equation}
of the wave-function $\Psi$. Indeed, once the magnetic kinetic energy and the interaction energy are assumed to be fixed by the form~\eqref{eq:Ker} of the wave function, only the potential energy can vary non-trivially. In this paper we will be interested in studying the ground state energy 
\begin{equation}\label{eq:intro gse}
E (N) := \inf\left\{ \E [\Psi], \: \Psi \in \Ker, \left\Vert \Psi \right\Vert_{L ^2 (\R ^2)} = 1 \right\}. 
\end{equation}
In introducing this variational problem, which seems to be of a novel type, we have two main physical motivations in mind:
\begin{itemize}
\item In usual quantum Hall bars, electrons are confined and the disorder potential is crucial for the understanding of the physics of the quantum Hall effect. In Laughlin's original picture these aspects are neglected and only the translation-invariant case $V\equiv 0$ is considered. The finite size of the system is taken into account by fixing the filling factor, which amounts to imposing a maximum degree to the polynomial (holomorphic) part of the wave-function. The disorder is neglected altogether, and it is argued that the proposed wave-function is robust enough that these simplifications do not harm the conclusions in real samples, a fact amply confirmed by experiments (see~\cite{STG} and references therein). It is this fact that we wish to study by considering the variational problem described above. In this context, the pair-interaction is given by the 3D Coulomb kernel $w(\xbf-\mathbf{y}) = |\xbf-\mathbf{y}| ^{-1}$, so the Laughlin state is not an exact ground state. We shall, however, assume that 
the interaction is negligible when we work with functions of $\Ker$.
\item It has been recognized for some time now  (\cite{Vie} and references therein) that the bosonic Laughlin state can in principle be created by fast rotation. Exploiting the analogy between the Coriolis and the Lorentz force, one may in this case identify the rotation frequency $\Omega$ around the axis perpendicular to the plane with an artificial magnetic field $B$. The potential $V$ is then 
$$V(\xbf) = V_{\rm trap} (\xbf) - \Om ^2 |\xbf| ^2$$
where $V_{\rm trap}$ is the trapping potential confining the atoms, which is corrected to take the centrifugal force into account. For cold atomic gases, the interactions are short range and can be effectively modeled by contact potentials, for which the bosonic Laughlin state is an exact, zero-energy eigenstate~\cite{PB,LS,LSY,RSY2}.  Experiments in these highly versatile systems (still elusive to this date) would allow an unprecedented direct probe of the properties of Laughlin's wave function. Recent experimental proposals~\cite{MF,RRD,Vie} involve engineering  the trapping potential, which can lead to new physics~\cite{RSY1,RSY2}. The variational problem~\eqref{eq:intro gse} is then intended as a way to study the shape imposed on the quantum Hall droplet by the trapping potential.
\end{itemize}

Motivated in particular by the two situations described above, our aim  is to investigate the \textit{incompressibility} properties of the Laughlin phase, in the form of a strong rigidity of its response to external potentials. Taking for granted the reduction to the LLL and cancellation of the interactions by the vanishing of the wave function along the diagonals of configuration space, we wish to see whether the Laughlin state, or a close variant, emerges as the natural ground state in a given potential. In particular, it is of importance to investigate the robustness of the correlations of Laughlin's wave function when the trapping potential is varied. 

In view of~\eqref{eq:Ker} and~\eqref{eq:intro gse}, we are looking for a wave-function of the form
$$
\Psi_F (z_1,\ldots,z_N)= c_F \PsiLau (z_1,\ldots,z_N) F (z_1,\ldots,z_N)
$$
where $F\in\BargN$ and $c_F$ is a normalization factor. A natural {\it guess} is that, whatever the one-body potential, the correlations stay in the same form and the ground state is well-approximated by a wave function 
\begin{equation}\label{eq:intro guess}
\Psi (z_1,\ldots,z_N) = c_{f_1} \PsiLau (z_1,\ldots,z_N) \prod_{j=1} ^N f_1 (z_j)  
\end{equation}
where the additional holomorphic factor $F$ that characterizes functions of $\Ker$ is uncorrelated, which is the meaning of the ansatz $F= f_1 ^{\otimes N}$. Note that, although this is a natural guess (in the absence of interactions it does not seem favorable to correlate the state more than necessary), it is {\it far from being trivial}. Indeed, although~\eqref{eq:start energy} is a one-body functional in terms of $\Psi$, the correlation factor $F$ really sees an effective, complicated many-body Hamiltonian because of the factor $\PsiLau$ in $\Psi_F$.

The energy functional~\eqref{eq:start energy} is of course very simple and all the difficulty of the problem lies in the intricate structure of the variational set~\eqref{eq:Ker} of fully-correlated wave-functions. The expected rigidity of the strongly correlated states of~\eqref{eq:Ker} 
should manifest itself through the property that their densities are essentially bounded above by a universal constant
\begin{equation}\label{eq:intro incomp bound}
\rhoP \lessapprox \frac{1}{\pi\ell N} \mbox{ for any } \Psi \in \Ker. 
\end{equation}
This is the incompressibility notion we will investigate, in the limit $N\to \infty$. In view of existing numerical computations of the Laughlin state (e.g.~\cite{Cif}),~\eqref{eq:intro incomp bound} can hold only in some appropriate weak sense, see below. We are not able at present to study the full variational set~\eqref{eq:Ker} and we will make assumptions on the possible form of the additional correlation factor $F$ in order to obtain a tractable model.

In fact we shall pursue along Laughlin's original intuition~\cite{Lau} and assume that particles are correlated only pairwise. This means that $F$ contains only two-body correlation factors, i.e. that it can be written in the form
\begin{equation}\label{eq:choice correlations}
F (z_1,\ldots,z_N) = \prod_{j = 1} ^N f_1 (z_j)\prod_{( i,j ) \in \{1,\ldots,N \} } f_2 (z_i,z_j)
\end{equation}
with $f_1$ and $f_2$ two polynomials ($f_2$ being in addition a symmetric function of $z,z'$) satisfying
\beq\label{eq:degree} \deg (f_1) \leq D N, \quad \deg(f_2) \leq D\eeq
for some constant $D$ independent of $N$. The degree of $f_2(z,z')$ is here understood as the degree of the polynomial in $z$ with $z'$ fixed (and vice versa).  Assumptions~\eqref{eq:choice correlations} and~\eqref{eq:degree} are, of course, restrictive, but still cover a huge class of functions with possibly very intricate correlations and leads to a problem that is not at all trivial. They can also be relaxed to some extent, see Section~\ref{sec:extensions}.

Our {\it main result} is that any normalized state of $\Ker$
\begin{equation}\label{eq:general state}
\Psi_F = c_F \, F\,  \PsiLau, \quad \Vert \Psi \Vert_{L ^2 (\R ^2)} = 1 
\end{equation}
corresponding to such correlation factors $F$ satisfies~\eqref{eq:intro incomp bound} in the limit $N\to \infty$ in a suitable weak sense made clear below\footnote{In particular we rescale the functions to take into account the fact that the right-hand side of~\eqref{eq:intro incomp bound} vanishes in the large $N$ limit.}. 

We then present some applications of theses estimates for the study of the variational problem~\eqref{eq:intro gse}, restricted to states of the above form. For a large class of radial potentials  $V(|z|)$ we confirm the optimality of the ansatz~\eqref{eq:intro guess}: adding a correlations factor $f_2$ in~\eqref{eq:choice correlations} cannot reduce the energy. If the potential is increasing, the Laughlin state is always preferred, i.e. $f_1\equiv 1$ gives the optimal energy. If the potential has a maximum at the origin, a situation we considered in~\cite{RSY1,RSY2}, it is favorable to choose $f_1 (z)= z ^{m}$ and optimize over $m$. That one needs only use this form for a large class of radial potentials is an illustration of the expected rigidity of the Laughlin phase.

\medskip

\noindent \textbf{Remarks on terminology.}
We note that the notion of incompressibility we investigate is related to, but different from, the notion that there should be a gap in the energy spectrum above the energy of the Laughlin state. The latter notion depends on the Hamiltonian under consideration and it is not clear what it means when the Laughlin state is not an exact eigenstate, as e.g. in the case of Coulomb interaction. The two notions are in turn related to, but different from, the existence of a mobility gap.

In the sequel we will refer to states of $\Ker$ as ``fully-correlated'' since they include (at least) the strongly correlated factor $\PsiLau$, and that adding more correlations does not decrease the interaction energy in contact potentials.  States of the form~\eqref{eq:intro guess} constitute the ``Laughlin phase'' in our language since they include exactly the correlations of the Laughlin state~\eqref{eq:intro Laughlin}, the remaining factor describing i.i.d particles. Our goal is thus, assuming that states are fully correlated (which can be justified rigorously in some asymptotic regimes~\cite{LS,RSY1,RSY2}), to show that one can reduce to the Laughlin phase. We will call $F$ a ``correlation factor'' since it encodes eventual correlations added to those of the Laughlin state, and the associated fully-correlated state $\Psi = c_F F \PsiLau$ will always be understood to be normalized in $L^2$ thanks to the constant $c_F$. 

\medskip

\noindent\textbf{Acknowledgments.} N.R. thanks the \textit{Erwin Schr\"odinger Institute}, Vienna, for its hospitality. Part of this research was done while the authors were visiting the \textit{Institut Henri Poincar\'e}, Paris, and another part while the authors were visiting the \textit{Institute for Mathematical Science} of the National University of Singapore. We acknowledge interesting discussions with Michele Correggi and financial support from the ANR (project Mathostaq, ANR-13-JS01-0005-01) and the Austrian Science Fund (FWF) under project P~22929-N16.

\section{Main results}\label{sec:main results}

From now on we consider the Laughlin state~\eqref{eq:intro Laughlin} for a given, fixed integer $\ell$, along with the associated set~\eqref{eq:Ker}. As announced we focus on proving a particular case of incompressibility by considering special trial states. Recall the definition~\eqref{eq:BargN} of the $N$-body bosonic Bargmann space and define the set  
\begin{multline}\label{eq:var set 2 D}
\VD = \Big\{ F \in \Barg ^N \, : \, \mbox{ there exist } (f_1,f_2) \in \Barg \times \Barg ^2, \deg (f_1) \leq D N, \deg(f_2) \leq D, 
\\  F (z_1,\ldots,z_N) = \prod_{j = 1} ^N f_1 (z_j)\prod_{1\leq i < j \leq N } f_2 (z_i,z_j) \Big\}
\end{multline}
where we understand that $f_1$ and $f_2$ are polynomials, and we assume bounds on their degrees. Note that it is very natural to be less restrictive in the bound on the degree of $f_1$ than on the degree of $f_2$: taking some $f_2 (z,z') = f_1 (z) f_1 (z')$ with $\deg (f_1)\leq D/2$ we could construct a $f_1$ factor of degree $DN$ anyway.

Given $F\in \VD$ we define the corresponding fully-correlated state
\begin{equation}\label{eq:PsiF}
\Psi_F (z_1,\ldots,z_N)= c_F \PsiLau (z_1,\ldots,z_N) F (z_1,\ldots,z_N)
\end{equation}
where $c_F$ is a normalization factor ensuring $\Vert \Psi_F \Vert_{L^2 (\R ^2)} = 1$. It is a well-known fact~\cite{Lau,RSY2} that the one-particle density of the Laughlin state is approximately constant over a disc of radius $\sim \sqrt{N}$ and then quickly drops to $0$. It is thus natural to also consider external potentials that live on this scale, which amounts to scale space variables and consider the energy functional
\begin{equation}\label{eq:scale ener}
\E_N [\Psi] = (N-1) \intR V\left(\xbf\right) \rhoP \left(\sqrt{N-1} \: \xbf \right) 
\end{equation}
where $\Psi$ is of the form~\eqref{eq:PsiF} and $\rhoP$ is the corresponding matter density. Note the choice of normalization: in view of~\eqref{eq:intro density}, 
$$\intR (N-1)  \rhoP \left(\sqrt{N-1} \: \xbf \right) d\xbf = 1.$$
That we scale lengths by a factor $\sqrt{N-1}$ instead of $\sqrt{N}$ is of course irrelevant for large~$N$. It only serves to simplify some expressions in Section~\ref{sec:incomp main} below. 

We define the ground state energy in the set $\VD$ as 
\begin{equation}\label{eq:energy Nn}
\ED:= \inf \left\{ \E_N [\Psi_F], \; \Psi_F \mbox{ of the form } (\ref{eq:PsiF}) \mbox{ with } F \in \VD \right\}.
\end{equation}
We can now state our incompressibility result, in the form of a universal lower bound on~$\ED$:

\begin{theorem}[\textbf{Weak incompressibility estimate for the Laughlin phase}]\label{thm:incomp main}\mbox{}\\
Let $V\in C ^{2} (\R^2)$ be increasing at infinity in the sense that
\begin{equation}\label{eq:increase V}
\min_{|x|\geq R} V(\xbf)\to \infty\quad\text{for}\quad R\to\infty.
\end{equation}
Define the corresponding ``bath-tub energy'' 
\begin{equation}\label{eq:bath tub}
E_V (\ell):= \inf\left\{ \intR V \rho \: | \: \rho \in L^1 (\R ^2), \: 0 \leq \rho \leq \frac{1}{\pi \ell},\ \intR\rho=1 \right\}.
\end{equation}
Then 
\begin{equation}\label{eq:main incomp}
\liminf_{N\to \infty} \ED \geq E_V(\ell).  
\end{equation}
\end{theorem}

That the estimate~\eqref{eq:main incomp} holds for any regular potential\footnote{The regularity assumption can be relaxed a bit.}  increasing at infinity is a weak formulation of the bound~\eqref{eq:intro incomp bound}. Indeed, observe that the energy in the potential $V$ may not be lower than the minimal value~\eqref{eq:bath tub} in the set of densities that satisfy~\eqref{eq:intro incomp bound} (note the rescaling of length and density units). It is well-known that the latter minimum energy is attained for a $\rho$ saturating the imposed $L ^{\infty}$ bound, a fact that is usually refered to as the ``bath-tub principle'', see~\cite[Theorem 1.14]{LL}. This weak version of a maximal density  bound provided by Theorem~\ref{thm:incomp main} is a consequence of the  two crucial characteristics of the FQHE: the restriction to the lowest Landau level and the strong, Laughlin-like, correlations. To clarify this, we make the following remarks. 

\begin{remark}[Illustrative comparisons]\label{rem:comparisons}\mbox{}\\
How would the energy~\eqref{eq:scale ener} behave in less constrained variational sets ? Three interesting cases are worth considering:
\begin{itemize}
\item \textit{Particles outside the} LLL. Suppose the single-particle Hilbert space was the full space $L^2 (\R ^2)$ instead of the constrained space $\LLL$. The minimization is then of course very simple and we would obtain $E_N = \min V$ by taking a minimizing sequence concentrating around a minimum point of $V$. Another way of saying this is that there is no upper bound whatsoever on the density of generic $L^2$ wave-functions, which is of course a trivial fact.
\item \textit{Uncorrelated bosons in the} LLL. For non interacting bosons in a strong magnetic field, one should consider the space $\bigotimes_s ^{N} \LLL$, the symmetric tensor product of $N$ copies of the LLL. The infimum in~\eqref{eq:scale ener} can then be computed considering uncorrelated trial states of the form $f ^{\otimes N}$, $f\in \LLL$. LLL functions do satisfy a kind of incompressibility property 
because they are of the form holomorphic $\times$ gaussian. This can be made precise by the inequality~\cite{ABN2,Car,LS}
\begin{equation}\label{eq:LLL incomp}
\sup_{z\in \C} \left|f (z) e^{-|z| ^2 / 2 }\right| ^2 \leq  \left\Vert f(.) e ^{-|\:.\:|/2} \right\Vert_{L ^2 (\R ^2)} ^2. 
\end{equation}
This is a much weaker notion however. (Compare~\eqref{eq:intro incomp bound} and~\eqref{eq:LLL incomp}, which leads to $\rhoP \leq 1$.) In our scaled variables, one may still construct a sequence of the form $f ^{\otimes N}$ concentrating around a minimum point of $V$ without violating~\eqref{eq:LLL incomp}. The liminf in~\eqref{eq:main incomp} is thus also equal to $\min V$  in this case.  
\item \textit{Minimally correlated fermions in the} LLL. Due to the Pauli exclusion principle, fermions can never be uncorrelated: the corresponding wave functions have to be antisymmetric w.r.t. exchange of particles: 
\begin{equation}\label{eq:Pauli}
\Psi(\xbf_1,\ldots,\xbf_i,\ldots,\xbf_j,\ldots,\xbf_N) = -\Psi(\xbf_1,\ldots,\xbf_j,\ldots,\xbf_i,\ldots,\xbf_N). 
\end{equation}
For  LLL wave functions (which are continuous) this implies 
$$\Psi(\xbf_i=\xbf_j) = 0 \mbox{ for any } i\neq j,$$
i.e. the wave function vanishes on the diagonals of the configuration space. Due to the holomorphy constraint, any $N$-body LLL fermionic wave-function is then of the form~\eqref{eq:PsiF}, where the Laughlin state is~\eqref{eq:intro Laughlin} with $\ell = 1$ and $F$ has bosonic symmetry. If $F = f ^{\otimes N}$, one could then talk of ``minimally correlated'' fermions\footnote{The Slater determinants for the Fermi sea at fixed total angular momentum can be obtained this way by taking $f(z) = z ^m,m\in \N$. Note that for $\ell = 1$ the Laughlin state is a Vandermonde determinant and hence describes free fermions.}. 

This case is covered by our theorem, and one obtains $E_V(1)$ as a lower bound to the energy. Observe that adding stronger correlations than those imposed by the Pauli principle, i.e. choosing an odd $\ell \geq 3$ increases the energy to the value $E_V(\ell)$ which is in general strictly larger than $E_V(1)$.
\end{itemize}
\hfill\qed
\end{remark}

A very natural question is that of the optimality of Theorem~\ref{thm:incomp main}. We can show that the lower bound~\eqref{eq:main incomp} is in fact optimal when the potential is radial, increasing or mexican-hat-like, the infimum being in fact asymptotically reached in the Laughlin phase~\eqref{eq:intro guess}. 

\begin{corollary}[\textbf{Optimization of the energy in radial potentials}]\label{cor:radial}\mbox{}\\
Let $V:\R ^2 \mapsto \R$ be as in Theorem~\ref{thm:incomp main}. Assume further that $V$ is radial, has at most polynomial growth at infinity,  and satisfies one of the two following assumptions 
\begin{enumerate}
\item $V$ is radially increasing,
\item $V$ has a ``mexican-hat'' shape: $V(|\xbf|)$ has a single local maximum at the origin and a single global minimum at some radius $R$.
\end{enumerate}
Then for $D$ large enough it holds 
\begin{equation}\label{eq:incomp opt}
\lim_{N\to \infty} \ED = E_V(\ell).  
\end{equation}
More precisely, in case $(1)$ 
\begin{equation}\label{eq:trial increas}
\E_N [\PsiLau] \to  E_V(\ell) \mbox{ when } N\to \infty
\end{equation}
and in case $(2)$ one can find a fixed number $\mb \in \R$ and a sequence $m(N) \in \N$ with $m \sim \mb N $ as $N\to \infty$ such that, defining 
\begin{equation}\label{eq:trial mexican 1}
\Psi_m (z_1,\ldots,z_N) := c_m \PsiLau (z_1,\ldots,z_N) \prod_{j=1} ^N z_j ^m 
\end{equation}
we have 
\begin{equation}\label{eq:trial mexican 2}
\E_N [\Psi_m] \to E_V(\ell) \mbox{ when } N\to \infty.
\end{equation}
\end{corollary}

What this corollary says is that, for the particular potentials under consideration, the simplified ansatz~\eqref{eq:intro guess} is optimal, at least amongst states built on $\VD$. We \emph{conjecture} that the latter restriction is unnecessary and that the result still holds when one considers the full variational set $\Ker$. Equations~\eqref{eq:trial increas} and~\eqref{eq:trial mexican 2} are consequences of results obtained in~\cite{RSY2} on the densities of the Laughlin state and the Laughlin $\times$ giant vortex states~\eqref{eq:trial mexican 2}: their densities in fact converge to the optimizer of the ``bath-tub'' energy~\eqref{eq:bath tub}. The remarkable fact about~\eqref{eq:trial increas} is that the Laughlin state stays an approximate minimizer in any radially increasing potential. No matter how steep and narrow
the imposed trapping potential, it is impossible to compress the Laughlin state while staying within the space $\Ker$. Also,~\eqref{eq:trial mexican 2} shows that if the 
potential has a radial well along some circle, the optimal density can be obtained by simply adding a vortex localized at the origin to the Laughlin state, and optimize over the added phase circulation. 

Let us stress that, if Corollary~\ref{cor:radial} shows that the incompressibility upper bound can be saturated by some special states, other states of $\Ker$ can in fact have a significantly lower maximal density. This is for example the case of the states~\eqref{eq:trial mexican 1}, when $m \gg N ^2$, see again~\cite{RSY2} for details on this claim.

An interesting question concerns the generalization of the above result to more general potentials. For the simple form of the potentials we chose, it is not too difficult to construct a trial state whose density converges to the minimizer of the bath-tub energy. Constructing a trial state that does the same for a more general potential remains an \emph{open problem}. We \emph{conjecture} that a suitable trial state can always be constructed in the form ~\eqref{eq:intro guess}, which would be a step towards confirming the robustness of Laughlin-like correlations in general potentials. 

Before turning to the proofs  of our results, let us comment on the physical relevance of the potentials considered in Corollary~\ref{cor:radial}. See also the analysis of FQH interferometers in~\cite{LFS}.

\begin{remark}[Rotating trapped Bose gases]\mbox{}\\
The shape of the potentials considered in Corollary~\ref{cor:radial} is inspired from the experimental situation in rotating Bose gases, for which the relevant choice of Laughlin state is~\eqref{eq:intro Laughlin} with $\ell = 2$. Here the potential is usually radial and as we explained before, of the form
$$V(\xbf) = V_{\rm trap} (|\xbf|) - \Om ^2 |\xbf| ^2.$$
In usual experiments  the potential is to a good approximation quadratic, $V_{\rm trap}(r) = \Om_{\rm trap}  ^2 r ^2$, and it can be shown\footnote{The stability of the system requires of course $V$ to be bounded below, and thus $\Om_{\rm trap} \geq \Om$.} that the Laughlin state is an exact ground state (in the case of a pure contact interaction). It seems, however, unlikely that the Laughlin state could be stabilized in such a potential, because the validity of the LLL approximation requires that $\Om_{\rm trap}  ^2 - \Om ^2$ be extremely small, and thus the residual trapping, modified by the centrifugal force, to be extremely weak (see in particular~\cite{RRD} for a thorough discussion of experimental issues). To provide a better confinement against centrifugal forces, it has been proposed to use a steeper potential, a popular proposal being
$$V_{\rm trap} (r) = \Om_{\rm trap}  ^2 r ^2 + k r ^4.$$
The effective potential taking the centrifugal force into account is then of the form 
\begin{equation}\label{eq:trap Bose gas}
V(|\xbf|) = \om N |\xbf| ^2 + k N ^2 |\xbf| ^4 
\end{equation}
with space variables scaled as in~\eqref{eq:scale ener} and $\om = \Om_{\rm trap}  ^2 - \Om ^2$. This potential is radial, increasing if $\om >0$ and mexican-hat like when $\om<0$. Corollary~\ref{cor:radial} applies directly in the case where $\om = \om_0 N ^{-1}$ and $k = k_0 N ^{-2}$ for fixed constants $\om_0,k_0$, and our methods can also accommodate more general choices. The previous results are thus a step towards the confirmation of a conjecture made in~\cite{RSY2}, that states of the form~\eqref{eq:trial mexican 1} asymptotically optimize the potential energy in potentials of the form~\eqref{eq:trap Bose gas}. This can, in principle, be checked experimentally (see again~\cite{RRD}): the states~\eqref{eq:trial mexican 1} all have a very rigid density profile, almost constant in a region and falling rapidly to $0$ elsewhere~\cite{RSY2}. This differs markedly from the behavior of the Bose condensed phase, whose density follows closely the shape of the trap (see e.g.~\cite{ABD,ABN1,
ABN2}), and may thus serve as an experimental probe of the Laughlin phase. More details on these aspects can be found in~\cite{RSY1,RSY2} and references therein.\hfill\qed
\end{remark}

\medskip

The rest of the paper contains the proofs of our main results. The core of the argument is in Section~\ref{sec:incompressibility} where we argue that wave-functions built on $\VD$, in fact satisfy more precise incompressibility estimates than stated in Theorem 2.1, with controlled error terms, see Theorem~\ref{theo:incomp two}. Our main tool is Laughlin's plasma analogy, recalled below, and a new approach to the mean-field limit of classical Gibbs states. The method is based on a theorem by Diaconis and Freedman recalled in Appendix~\ref{sec:DiacFreed}, and we believe that it is of independent interest. The proof of Theorem~\ref{thm:incomp main} is concluded in Section~\ref{sec:concl main}. The additional ingredients needed for the proof of Corollary~\ref{cor:radial}, taken from ~\cite{RSY2}, are recalled in Subsection~\ref{sec:cor proof}. A final Section~\ref{sec:extensions} explains possible generalizations of our results for correlation factors more complicated than~\eqref{eq:choice correlations}.

\section{Incompressibility estimates and the plasma analogy}\label{sec:incompressibility}

\subsection{Main estimate}\label{sec:incomp main}

In this section we consider a correlation factor $F\in \VD$, that is we pick two functions $f_1,f_2 \in \B ^1, \B ^2$ and take $F$ to be of the form
\begin{equation}\label{eq:F two correl}
F(z_1,\ldots,z_N) := \prod_{j=1} ^N f_1 (z_j) \prod_{1\leq i < j \leq N}  f_2 (z_i,z_j) (z_i-z_j) ^m. 
\end{equation}
Without loss (changing the value of the parameter $m\in \N$ if necessary), we may assume that~$f_2$ is not identically $0$ on the diagonal $z_i = z_j$, in which case the set $\{ f_2 (z_i,z_j ) =  0\} \cap \{z_i = z_j\}$ is of dimension~$0$ (consists only of isolated points). Since the  normalized wave-function~\eqref{eq:general state} does not change when the polynomials $f_1,f_2$ are multiplied by constants, we may without loss assume that 
\begin{align}\label{eq:bounds f1 f2}
|f_1 (z)| &\leq |z| ^{DN} + 1 \nonumber \\
|f_2 (z,z')| &\leq |z| ^{D} + |z'| ^{D} + 1 
\end{align}
in view of our assumptions on their degrees.

\medskip

Fixing some $\ell \in \N$, we want to analyze the density of the corresponding fully-correlated state~\eqref{eq:PsiF}, appearing in the energy functional~\eqref{eq:start energy}. To this end we use Laughlin's plasma analogy~\cite{BCR,Lau,Lau2,LFS}, regarding  $|\Psi_F| ^2$ as the Gibbs measure of a system of classical charged particles. We first scale variables by defining 
\begin{equation}\label{eq:defi Gibbs} 
\muF (Z) := (N -1) ^N \left| \Psi_F \left(\sqrt{N-1} \: Z\right) \right| ^2
\end{equation}
and we can then write $\muF$ as ($\ZF\in \R$ normalizes the function in $L ^1$)
\begin{equation}\label{eq:plasma analogy}
\muF (Z) = \frac{1}{\ZF} \exp\left( -\frac1T H_F (Z) \right) 
\end{equation}
where the temperature is 
$$T=(N-1) ^{-1}$$
and the Hamiltonian function $H_F$ is 
\begin{multline}\label{eq:class hamil}
H_F (Z) = \sum_{j=1} ^N \left( |z_j| ^2 - \frac{2}{N-1} \log | g_1 (z_j)|\right) 
\\ + \frac{2}{N-1} \sum_{1\leq i < j \leq N} \left( - \left( \ell + m \right)\log|z_i-z_j| - \log |g_2 (z_i,z_j) |\right)   
\end{multline}
with 
\begin{equation}\label{eq:scale correl factor}
g_1 (z) = f_1 \left(\sqrt{N-1} z \right), \quad g_2 (z,z') = f_2 \left(\sqrt{N-1} z, \sqrt{N-1} z' \right).
\end{equation}
This is a classical Hamiltonian with mean-field pair interactions: the factor $(N-1) ^{-1}$ multiplying the $2$-body part, due to the scaling in~\eqref{eq:defi Gibbs}, makes the interaction energy of the same order of magnitude as the one-body energy. This is crucial for extracting information from the plasma analogy in the limit $N\to \infty$. Here we use in  an essential manner the structure ``holomorphic $\times$ gaussian'' of LLL functions.

One may interpret~\eqref{eq:class hamil} as the energy of $N$ charged point particles located at $z_1,\ldots,z_N \in \R ^2$ under the influence of the following potentials:
\begin{itemize}
\item A confining harmonic electrostatic potential (the $|z_j| ^2$ term in the one-body part).
\item The potential generated by fixed charges located at the zeros of the function $g_1$ (the $-\log |g_1 (z_j)|$ term in the one-body part).
\item The usual $2$D Coulomb repulsion between particles of same charge (the $-(\ell + m) \log|z_i-z_j|$ terms in the two-body part). 
\item The Coulomb repulsion due to phantom charges located at the zeros of $g_2 (z_i,z_j)$ (the $- \log |g_2 (z_i,z_j) |$ terms in the two-body part). This term describes an intricate two-body interaction and it is responsible for most of the difficulties we encounter below. Particle $i$ feels an additional charge attached to each other particle $j$ in a complicated manner encoded in the zeros of $g_2 (z_i,z_j)$. When particle $j$ moves in a straight line, the phantom charges attached to it may follow any algebraic curve in the plane.
\end{itemize}

By definition, $\muF$ minimizes the free energy functional 
\beq\label{eq:free energy funct}
\F [\mu] := \E [\mu] + T \intRN \mu(Z) \log (\mu(Z)) dZ\eeq
with the energy term 
\beq\label{eq:energy funct} 
\E [\mu]:= \intRN H_F (Z) \mu(Z) dZ
\eeq
on the set $\PP (\R ^{2N})$ of probability measures on $\R ^{2N}$, and we have the relation
\begin{equation}\label{eq:partition}
\F [\muF] = \inf_{\mu \in \PP (\R ^{2N})} \F [\mu] = - T \log \ZF = - \frac{\log \ZF}{N-1}. 
\end{equation}
Our goal here is to relate, in the large $N$ limit, the $N$-body minimization problem~\eqref{eq:partition} to the analogous problem for the
mean-field energy functional defined as 
\begin{multline}\label{eq:MF func}
\MFf [\rho] := \int_{\R ^2} \left(|z| ^2 -  \frac{2}{N-1} \log |g_1 (z)|\right) \rho(z) dz  
\\ + \int_{\R ^4} \rho(z) \big( - \left( \ell+m \right)\log|z- z '| - \log |g_2 (z,z')|\big) \rho(z')dz dz'
\end{multline}
for a probability measure $\rho \in \PP (\R ^2)$. We shall denote by $\MFe,\rhoMF$ the minimum of the functional~\eqref{eq:MF func} and a minimizer respectively. Both may be proved to exist by standard arguments. More precisely, introducing for any symmetric $\mu\in\PPs(\R ^{2N})$ the $k$-marginals\footnote{Or reduced $k$-body densities, or $k$-correlation functions. Note that by symmetry, the choice of the variables over which to integrate is not important.}
\begin{equation}\label{eq:marginals}
\mu  ^{(k)} (z_1,\ldots,z_k):= \int_{\R ^{2(N-k)}} \mu (z_1,\ldots,z_k,z'_{k+1},\ldots,z'_{N}) dz'_{k+1} dz'_N
\end{equation}
we should expect that, when $N$ is large, 
\begin{equation}\label{eq:formal marginal one}
\muFone \approx \rhoMF 
\end{equation}
and more generally
\begin{equation}\label{eq:formal marginals}
\muFk \approx \left( \rhoMF \right) ^{\otimes k} 
\end{equation}
in some appropriate sense, which is essentially a consequence of an energy estimate of the form
\begin{equation}\label{eq:formal energy}
- T \log \ZF  \approx N \MFe.
\end{equation}
The rationale behind these expectations is that, $N$ being large and the pair-interactions in~\eqref{eq:class hamil} scaled to contribute at the same level as the one-body part, one should expect an uncorrelated ansatz $\muF \approx \rho ^{\otimes N}$ to be a reasonable approximation for the minimization
problem~\eqref{eq:partition}. Then, since the temperature $T\to 0$ in this limit, the entropy term can be safely dropped as far as leading order considerations are concerned. With this ansatz and simplification, the minimization of $\F[\rho ^{\otimes N}]$ reduces to that of $\MFf[\rho]$, so that one may expect~\eqref{eq:formal marginals} and~\eqref{eq:formal energy} to be reasonable guesses.

The $k=1$ version of~\eqref{eq:formal marginals} provides the incompressibility property we are after: As we will show below 
\begin{equation}\label{eq:up bound MF density}
\rhoMF \leq \frac{1}{\ell \pi}   
\end{equation}
independently of $F$. In fact, as we have shown in~\cite{RSY1,RSY2}, the density of the Laughlin state saturates this bound in the sense that when $F\equiv 1$
\[
\muFone \approx \varrho ^{\rm MF}  = \begin{cases} \displaystyle \frac{1}{\ell \pi} \mbox{ in } \supp (\varrho ^{\rm MF}) \\  0 \mbox{ otherwise.} \end{cases}
\]
We can thus be a little bit more precise by interpreting our results as saying that no additional correlation factor of the type considered can compress the Laughlin state. 

The result we aim at in this section is the following version of incompressibility.

\begin{theorem}[\textbf{Incompressibility for states with two-body correlations}]\label{theo:incomp two}\mbox{}\\
Let $\muF$ be defined as above. Pick some test one-body potential $U$ such that $U,\Delta U \in L ^{\infty} (\R ^2)$. For any $N$ large enough and $\eps>0$ small enough there exists an absolutely continuous probability measure $\rhoF \in L ^1 (\R ^2)$ satisfying
\begin{equation}\label{eq:bound MF thm}
\rhoF \leq \frac{1}{\pi(\ell + m)} +\eps \frac{\sup |\Delta U|}{4 \pi (\ell + m)}
\end{equation}
such that 
\begin{equation}\label{eq:incomp two}
\int_{\R ^2} U \muFone \geq  \int_{\R ^2} U  \rhoF - C (N\eps) ^{-1}  \error(m,D,\eps U)
\end{equation} 
where $\error(m,D,\eps U)$ is a quantity satisfying 
\begin{equation}\label{eq:bound error}
\left| \error(m,D,\eps U) \right| \leq  1  + (\ell + m + D) \log N + D\log D + (1+\eps\,{\rm sup}\, |\Delta U|). 
\end{equation}
\end{theorem}

Before we present the proof of this result, several remarks are in order:
\begin{itemize}
\item The $L^1$ function $\rhoF$ appearing in the theorem is in fact the solution to a variational problem for a perturbed functional related to~\eqref{eq:MF func},  that we introduce below. The perturbation is responsible for the $\eps$-dependent term in the right-hand side of~\eqref{eq:bound MF thm}.
\item The upper bound~\eqref{eq:bound MF thm} is a slightly weakened version of~\eqref{eq:up bound MF density}. In applications there is a trade-off between choosing $\eps$ very small so that the upper bound~\eqref{eq:bound MF thm} is as close as possible to~\eqref{eq:up bound MF density}, and $\eps$ very large so that the error term in the lower bound~\eqref{eq:incomp two} is as small as possible.
\item One may wonder whether the lower bound~\eqref{eq:incomp two} is optimal. Our method yields an upper bound in terms of a slightly different reference density, indicating that~\eqref{eq:incomp two} is in fact optimal up to remainder terms, see Remark~\ref{rem:converse incomp} below.   
\item We have made error terms explicit but what matters for the sequel is that 
$$ (N\eps) ^{-1} \error(m,D,\eps U) \to 0$$
when $N\to \infty$ and $1\gg \eps \gg N ^{-1}$ with $m$,  $D$  and $U$ being kept fixed.
\end{itemize}
We note also the following additions to Remark~\ref{rem:comparisons}.

\begin{remark}[Effect of correlations]\label{rem:effect correl}\mbox{}\\
One may be surprised that the marginal densities of highly correlated trial states are efficiently approximated in the form $\rho ^{\otimes k}$, which corresponds to independent and identically distributed classical particles. There are several answers to this. First, as noted in Remark~\ref{rem:comparisons}, the correlations are responsible for reducing substantially the maximum density. Then, one should note that although the marginal densities of $|\Psi_F| ^2$ are safely approximated by $\rho ^{\otimes k} $, the appropriate $\rho$ is in general not the square of the modulus of a LLL function. The emergence of the appropriate density profile is thus also a consequence of the correlations. Finally, the fact that the density $|\Psi_F| ^2$ factorizes is by no means an indication that the state itself is effectively uncorrelated. Indeed, much of FQHE physics is based on the correlations in the phase of the wave-function. It is of crucial importance~\cite{STG,Jai} that the correlations may be seen as due to 
each particle carrying quantized vortices, and thus phase singularities.\hfill \qed
\end{remark}

The rest of this section contains the proof of Theorem~\ref{theo:incomp two} and related comments. The main idea is as follows : We want to investigate the large $N$-limit of the minimization problem~\eqref{eq:partition} with the goal of justifying~\eqref{eq:formal marginals} and~\eqref{eq:formal energy}. This may look like a classical problem at first, but a little further thought reveals that existing methods are not enough to deal with it at the level of precision we require. First, note that the data of the problem, namely $g_1$, $g_2$ and $m$ may depend on $N$, which rules out methods based on compactness arguments~\cite{MS,Kie1,CLMP,KS}. Secondly, the two-body potential is not of the form $w(z-z')$ and is not of positive type\footnote{i.e. of positive Fourier transform.}, which is the crucial property on which known quantitative methods rely (see~\cite{RSY2,RS,SS} and references therein). 

Our approach is based on a quantitative version of the  Hewitt-Savage~\cite{HS} theorem. The theorem we use, due to Diaconis and Freedman~\cite{DF}, is recalled in Appendix~\ref{sec:DiacFreed} for the convenience of the reader. It states that the $k$-marginals of any classical symmetric $N$-body state are close to convex combinations of uncorrelated states of the form $\rho ^{\otimes k}$ when $k\ll \sqrt{N}$. In the regime we investigate, the temperature is so small that we may neglect the entropy term in~\eqref{eq:free energy funct}. The free-energy functional~\eqref{eq:free energy funct} is thus essentially an affine functional of $\mu$ and the infimum is approximately reached on the extremal points of the convex set $\PPs (\R ^{2N})$. But the Diaconis-Freedman theorem essentially says that the marginals of the latter are all close to states of the form $\rho ^{\otimes  k},\: \rho \in \PP (\R ^2)$, which explains why~\eqref{eq:formal energy} should hold in reasonable generality. 

\subsection{Preliminaries}\label{sec:prelim incomp}

Our 
strategy requires two modifications of the Hamiltonian:
\begin{itemize}
\item First, the Diaconis-Freedman theorem can be put to good use only with two-body Hamiltonians that are not singular along the diagonals $z_i=z_j$. 
Our Hamiltonian~\eqref{eq:class hamil} has Coulomb singularities, however, that we regularize by smearing out charges.
The charge distribution of a unit charge smeared over a disc of radius $\alpha$ is
\beq\label{eq:deltaalpha}
\delta_\alpha(z):=\begin{cases}
             \displaystyle \frac 1{\pi\alpha^2} &\mbox{ if } |z|\leq \alpha \\
             0 &\mbox{ if } |z|> \alpha.
            \end{cases}
\eeq
It tends to the delta function as $\alpha\to 0$. The corresponding electrostatic potential is
\begin{equation}\label{eq:log alpha}
- \log_\alpha| z|  := -\int \log|z-w|\, \delta_\alpha(w) dw .\end{equation}
Since $\Delta \log|z|=2\pi\delta(z)$ we have 
\beq \label{eq:delta_logalpha}
\Delta \log_\alpha|z|=2\pi\delta_\alpha(z).
\eeq
Using Newton's theorem (\cite{LL}, Theorem 9.7 or~\cite{RSY2}, Lemma 3.3) it is straightforward to compute
\begin{equation}\label{eq:log alpha 2}
- \logal |z| = \begin{cases}
             \displaystyle -\log\alpha+\half \left(1-\frac{|z|^2}{\alpha^2}\right) &\mbox{ if } |z|\leq \alpha \\
             - \log |z| &\mbox{ if } |z|\geq \alpha
            \end{cases}
\end{equation}
for any $z\in \R ^2$. Clearly we have 
\begin{equation}\label{eq:log dom logal}
-\log_\alpha |z|\leq \max\left( -\log |z|, -\log \alpha + \half\right)
\end{equation}
and $-\logal |z|$ tends pointwise to $-\log |z|$ for $z\neq 0$ as $\alpha\to 0$. Moreover,~\eqref{eq:log alpha 2} implies
\begin{equation}\label{eq:log alpha 3}
- \frac d{d\alpha}\logal |z| = \begin{cases}
             \displaystyle -\frac 1\alpha+\frac {|z|^2}{\alpha^3} &\mbox{ if } |z|\leq \alpha\\
             0 &\mbox{ if } |z|\geq \alpha.
            \end{cases}
\end{equation}

To obtain energy lower bounds we will replace the two-body potential 
$$ W(z,z') = -2(\ell+m) \log |z-z'| - 2\log|g_2 (z,z')|$$
appearing in~\eqref{eq:class hamil} by the regularized
\begin{equation}\label{eq:reg two body}
\Wal(z,z'):= - 2(\ell+m) \logal |z-z'| - 2 \log_\alpha \left| g_2 (z,z')\right|, 
\end{equation}
recalling that, by~\eqref{eq:log alpha},  
\beq \log_\alpha \left| g_2 (z,z')\right|=\int \log \left| g_2 (z,z')-w\right|\,\delta_\alpha(w)\,dw.\eeq 
We shall also use that the $\log$ of the absolute value of a holomorphic function is subharmonic, so 
\beq \Delta\log_\alpha|g_1(z)|\geq 0\quad\hbox{ and }\quad\Delta_z\log_\alpha|g_2(z,z')|\geq 0.\eeq
 Finally, for $\alpha=0$ we define 
\beq \log_0=\log \quad\hbox{ and }\quad \delta_0(z)=\delta(z).\eeq

\medskip

\item The second modification of the Hamiltonian is due to the following. To show that~\eqref{eq:formal marginal one} follows from~\eqref{eq:formal energy}, the usual way is to vary the one-body potential in the Hamiltonian and use the Feynman-Hellmann principle to obtain a weak-$\ast$ convergence in $L^1$. Here we are looking for quantitative estimates on objects that might depend on $N$ even in the mean-field approximation, so weak convergence arguments are not available. We thus implement the variation explicitly, adding a weak one-body potential $\eps U$ to the Hamiltonian, with $\eps$ a (small) number and $U \in L ^{\infty} (\R ^2)$ being the potential appearing in Theorem~\ref{theo:incomp two}. Analyzing the mean-field limits at $\eps = 0$ and $\eps\neq 0$ small enough will give information on the one-body density at $\eps = 0$, i.e. for the original problem.
\end{itemize}

These considerations lead us to the following modified Hamiltonian:
\begin{multline}\label{eq:modif hamil}
\Hepal (z_1,\ldots,z_N) : = \sum_{j=1} ^N \left( |z_j| ^2 - \frac{2}{N-1} \log | g_1 (z_j)| + \eps U (z_j) \right) \\
+ \frac{1}{N-1} \sum_{1\leq i < j \leq N} \Wal(z_i,z_j) 
\end{multline}
to which we associate the free-energy functional (still with $T = (N-1) ^{-1}$ and $\mu \in \PP (\R ^{2N} )$)
\begin{equation}\label{eq:free ener perturb}
\Fepal [\mu] := \int_{Z\in \R ^{2N}} \Hepal (Z) \mu (Z) dZ + T  \int_{Z\in \R ^{2N}} \mu (Z) \log (\mu(Z)) dZ. 
\end{equation}
We denote  by
\[
\muepal (Z) := \frac{1}{\Zepal} \exp\left(-\frac1T \Hepal (Z) \right)
\]
the associated Gibbs measure, and we have
\[
\Fepal [\muepal] = - T \log \Zepal = \inf_{\mu \in \PP (\R ^{2N})} \Fepal [\mu]. 
\]
We will relate the minimization of~\eqref{eq:free ener perturb} to that of the mean-field energy functional
\begin{equation}\label{eq:modif MFf}
\MFfepal [\rho] :=  \int_{\R ^2} \left(|z| ^2 - \frac{2}{N-1} \log |g_1 (z)| + \eps U(z) \right) \rho(z) dz  \\
+ \half \int_{\R ^4} \rho(z) \Wal(z,z') \rho(z')dz dz'
\end{equation}
with minimum $\MFeepal$ and minimizer $\rhoMFepal$ (amongst probability measures). With this notation, the original problem we are ultimately interested in corresponds to the choice $\ep=\al = 0$.

\medskip

As a preparation for the study of the MF limit below, we investigate the properties of this family of mean-field functionals. In particular we prove the upper bound~\eqref{eq:up bound MF density}, and show that some weak version of it survives the introduction of the small parameters $\al$ and~$\eps$. 

\begin{lemma}[\textbf{The mean-field densities}]\label{lem:MF func}\mbox{}\\
For any $\eps,\al$ small enough there exists at least one $\rhoMFepal$ minimizing~\eqref{eq:modif MFf} amongst probability measures. Any minimizer satisfies the variational equation
\begin{align}\label{eq:MF var eq}
\int_{\R ^2} \Wal (z,z') \rho (z') &dz' + \left( |z| ^2 - \frac{2}{N-1} \log |g_1 (z)| + \eps U \right) \nonumber
\\ &{\begin{cases}= \MFeepal + \half \int_{\R ^4} \rhoMFepal(z) \Wal(z,z') \rhoMFepal (z')dz dz' \mbox{ if  } z \in \supp (\rhoMFepal)\\{}
\\ \geq \MFeepal + \half \int_{\R ^4} \rhoMFepal(z) \Wal(z,z') \rhoMFepal (z')dz dz'  \mbox{ otherwise.}\end{cases}}
\end{align}
Moreover, we have the following bounds: 
\begin{itemize}
\item (Case $\alpha = 0$). For any $z\in \R^2$
\begin{equation}\label{eq:bound MF}
\rhoMF_{\ep,0} (z) \leq \frac{1}{\pi(\ell+m)} + \eps \frac{\sup |\Delta U|}{4 \pi (\ell+m)}
\end{equation}
\item (Case $\alpha > 0$). For any $z\in \R^2$
\begin{equation}\label{eq:bound MF perturb}
\delta_\alpha*\rhoMF_{\ep,\alpha} (z) \leq \frac{1}{\pi(\ell+m)} + \eps \frac{\sup |\Delta U|}{4 \pi (\ell+m)}
\end{equation}
\end{itemize}
\end{lemma}

Inequality~\eqref{eq:up bound MF density} in the case $F\equiv 1$ is simply~\eqref{eq:bound MF} with $m=0$ and $\eps = 0$.

\begin{proof}
The existence part and the variational equation follow by standard methods, see~\cite[Chapter 1]{ST}. In particular, note that the confining potential guarantees the tightness of minimizing sequences. The bounds in~\eqref{eq:bound MF}  and~\eqref{eq:bound MF perturb} 
are deduced by applying the Laplacian to the variational equation~\eqref{eq:MF var eq}, using that 
$$\Delta\log_\alpha|z-z'|=2\pi\delta_\alpha(z-z'), \; \Delta |z|^2=4$$ 
and 
$$\Delta\log_\alpha|g_1(z)|\geq 0, \; \Delta_z\log_\alpha|g_2(z,z')|\geq 0.$$
The latter inequalities follow from the fact that $g_1$ and $g_2$ are holomorphic functions. 
\end{proof}

In view of the properties of our regularized interaction potentials, it is clear that $\MFeepal$ approximates $\MFe_{\eps,0}$ in the limit $\alpha \to 0$, which is the content of the following:

\begin{lemma}[\textbf{Small $\alpha$ limit of the mean-field energy}]\label{lem:alpha to 0}\mbox{}\\
For any $\eps,\al$ small enough the following bound holds:
\begin{equation}\label{eq:MFe alpha to 0}
\MFeepal \leq \MFe_{\eps,0} \leq \MFeepal + C \left(1+\eps \sup | \Delta U |\right)  \,\alpha ^{\min(2,2/D)}
\end{equation}
where $D$ is the degree of $g_2(z,z')$ as a function of either $z$ or $z'$.
\end{lemma}

\begin{proof}
The lower bound $\MFeepal \leq \MFe_{\eps,0}$ is an obvious consequence of $-\logal |z|\leq -\log |z|$. 
For the upper bound we use the formula
\begin{equation}
\MFeepal-\MFe_{\eps,0}=\int_0^\alpha \frac d{d\alpha'} \MFe_{\eps,\alpha'}d\alpha'.
\end{equation}
We evaluate $\frac d{d\alpha'} \MFe_{\eps,\alpha'}$ by means of the Feynman-Hellman theorem:
\begin{equation}
\frac d{d\alpha'} \MFe_{\eps,\alpha'}=
\left [\frac d{d\alpha'}\mathcal E^{\rm MF}_{\eps,\alpha'}\right][\rho^{\rm MF}_{\eps,\alpha'}],  
\end{equation}
and, for any $\rho$,
\begin{equation}\label{2.32}
\left [\frac d{d\alpha'}\mathcal E^{\rm MF}_{\eps,\alpha'}\right ][\rho]=\frac 12 \iint \rho(z)\frac d{d\alpha'}W_{\alpha'}(z,z')\rho(z') dzdz'.
\end{equation}
Moreover, by
\eqref{eq:log alpha 3} we have
\begin{equation}\label{2.33}
\left|\frac d{d\alpha'}W_{\alpha'}(z,z')\right| \leq \frac 1{\alpha'}(\ell+m+1)\mbox{ if } |z-z'|\leq \alpha' \mbox{ or } 
 |g_2(z,z')|\leq \alpha'
\end{equation} 
and $0$ otherwise. 

Next we note that as a consequence of~\eqref{eq:bound MF perturb}, we have for any set $\Lambda $ of area $|\Lambda|$ that can be covered by $O(|\Lambda|/\alpha^2)$ discs of radius $\alpha$
\begin{equation}\label{bound_on_rho} 
\int_\Lambda\rho^{\rm MF}_{\eps,\alpha}(z)dz\leq C\,\frac{1+\eps \sup | \Delta U |}{\ell+m}\, |\Lambda|.\end{equation}
Since $g_2(z,z')$ has a factorization
\beq\label{factorization} g_2(z,z')=\prod_{i=1}^D(z-c_i(z'))\eeq
where the zeros $c_i(z')$ are algebraic functions of $z'$,  it is easy to see that
for any $z'$ the total area of the set of $z$'s where $|z-z'|\leq \alpha' \mbox{ or } |g_2(z,z')|\leq \alpha'$ is bounded by 
$$
\begin{cases}
C (\alpha')^{2/D} \mbox { if } D \geq 1 \\
C (\alpha') ^2 \mbox { if } D = 0 
\end{cases}
$$ 
where $C$ and $D$ depend only on $g_2$ but are independent of $z'$. Moreover this set may be covered by $O(\alpha^{2/D-2})$ discs of radius $\alpha ^2$ in the first case and $O(1)$ discs in the second case. We may then employ~\eqref{bound_on_rho} on it.

 Combining~\eqref{2.32}, ~\eqref{2.33},~\eqref{bound_on_rho},
and the normalization $ \int \rho_{\eps,\alpha'}^{\rm MF}=1$, we obtain
\beq \left|\frac d{d\alpha'} \MFe_{\eps,\alpha'}\right|\leq
C\left(1+\eps \sup |\Delta U | \right)(\alpha')^{\min(2,2/D)}.
\eeq
Integrating from 0 to $\alpha$ then gives~\eqref{eq:MFe alpha to 0}.

\end{proof}

\subsection{Mean-field/small temperature limit}\label{sec:MF lim}

The crux of our method is the proof of the following lower bounds for our family of many-body free energies. The corresponding upper bounds are easy to derive, using the usual $\rho ^{\otimes N}$ ansatz, but they shall not concern us at this stage. 

\begin{proposition}[\textbf{Free-energy lower bounds}]\label{pro:estim free ener}\mbox{}\\
Under the previous assumptions, the following holds for any $\alpha,\eps$ small enough and $\mu \in \PP (\R^{2N})$:
\begin{align}\label{eq:free ener MF}
\Fepal [\mu] &\geq N \MFeepal - C \left(\ell+m+D \right)\left(|\log \alpha| + 1 \right) \nonumber \\
&-C \left( D \log D + D \log N+ (\ell + m)\log (\ell + m ) \right).
\end{align}
\end{proposition}

\begin{proof}
In view of the symmetry of the Hamiltonian~\eqref{eq:modif hamil}, we may consider the minimization restricted to symmetric probabilities. For such a $\mu \in \PPs (\R ^{2N})$ the free-energy~\eqref{eq:free ener perturb} may be rewritten
\begin{multline*}
\Fepal [\mu] = N \int_{\R  ^2} \mu ^{(1)} (z) \left( |z| ^2 - \frac{2}{N-1} \log | g_1 (z)| + \eps U(z) \right) dz \\ 
+ \frac{N}{2}  \int_{\R  ^4} \mu ^{(2)} (z,z') \Wal(z,z') dz dz' 
+ \frac{1}{N-1}  \int_{Z\in \R ^{2N}} \mu (Z) \log (\mu(Z)) dZ
\end{multline*}
with $\mu ^{(1)}$ and $\mu ^{(2)}$ the first and second marginals of $\mu$, defined as in~\eqref{eq:marginals}. We first deal with the entropy term: by positivity of relative entropies (see e.g. Lemma 3.1 in~\cite{RSY2})
\[
\int_{Z\in \R ^{2N}} \mu (Z) \log (\mu(Z)) dZ \geq  \int_{Z\in \R ^{2N}} \mu (Z) \log (\nu(Z)) dZ
\]
for any probability measure $\nu$. We set
\[
\nu_0 (z) = c_0 \exp(-|z| ^2 ) 
\]
with $c_0$ a normalization constant and apply the above inequality with $\nu = \nu_0 ^{\otimes N}$. Integrating over $N-1$ variables we then obtain 
\begin{align*}
\int_{Z\in \R ^{2N}} \mu (Z) \log (\mu(Z)) dZ &\geq  N \int_{z\in \R ^{2}} \mu ^{(1)} (z) \log (\nu_0(z)) dz 
\\& = N \log (c_0) - N \int_{z\in \R ^{2}} |z| ^2\mu ^{(1)} (z) dz 
\end{align*}
which gives the lower bound 
\begin{align}\label{eq:low bound Fepal 1}
\Fepal [\mu] &\geq N \int_{\R  ^2} \mu ^{(1)} (z)  \left( (1-N ^{-1}) |z| ^2  - \frac{2}{N-1} \log | g_1 (z)| + \eps U(z)  \right)  dz \nonumber
\\&+ \frac{N}{2}  \int_{\R  ^4} \mu ^{(2)} (z,z') \Wal(z,z') dz dz' - C 
\end{align}
for any $\mu \in \PPs (\R ^{2N})$. Now we apply to $\mu$ the construction of~\cite{DF} recalled in Appendix~\ref{sec:DiacFreed}. This gives a $P_{\mu} \in \PP (\PP (\R ^2))$, a Borel probability measure over the probability measures of $\R ^2$, and a 
\begin{equation}\label{eq:mut}
\mut := \int_{\rho \in \PP (\R ^2)} \rho ^{\otimes N} P_{\mu} (d\rho) \in \PPs (\R ^{2N})
\end{equation}
such that, using~\eqref{eq:marginals DF},
\begin{align}\label{eq:mut two}
\mu ^{(1)} (z) &= \mut ^{(1)} (z) \nonumber
\\ \mu  ^{(2)}(z,z') &= \frac{N}{N-1} \mut ^{(2)} (z,z') - \frac{1}{N-1} \mut ^{(1)}(z) \delta(z-z')\nonumber \\& =\tilde\mu^{(2)}(z,z')+\frac 1{N-1}\left[\tilde\mu^{(2)}(z,z')-\mut ^{(1)}(z) \delta(z-z')\right]. 
\end{align}
It is at this point that it proves useful to have regularized the Coulomb part of the two-body interaction: since $\Wal$ is locally bounded from above we can insert~\eqref{eq:mut two} in~\eqref{eq:low bound Fepal 1}. We then obtain
\begin{align}\label{eq:low bound Fepal 2}
\Fepal[\mu] &\geq  N\left[ \int_{\R^2} \mut ^{(1)} (z)  \left( (1-2N ^{-1})|z| ^2  - \frac{2}{N-1} \log | g_1 (z)| + \eps U(z)  \right)  dz \nonumber
\right.\\& \left.+ \frac12  \int_{\R  ^4} \mut ^{(2)} (z,z') \Wal(z,z') dz dz'\right] \nonumber
\\&+ \int_{\R  ^4}\left[ (|z|^2+|z'|^2)/2+\frac{N}{2(N-1)} \Wal(z,z')\right]\mut ^{(2)} (z,z') dz dz'\nonumber
\\&  - \frac N{2(N-1}\int_{\R ^2} \mut ^{(1)} (z) \Wal(z,z)-C. %
\end{align} 
The last two lines are error terms, and in the second to last line we have borrowed 
\beq \int_{\R^2} |z|^2\mut ^{(1)} (z)dz=\half\int_{\R  ^4}(|z|^2+|z'|^2)\mut ^{(2)} (z,z')dz dz'\eeq
 from the main term in the first two lines; this account for the $-2N^{-1}|z|^2$ rather than $-N^{-1}|z|^2$ in the main term. 
 
 The error term in the second to last line is easily estimated, using ~\eqref{eq:bounds f1 f2} and~\eqref{eq:log alpha 2}:
\begin{align}\label{eq:low bound pot}
\Wal (z,z') &\geq - 2 (\ell + m ) \log \left(|z| + |z'|\right) - 2 \log \left( N ^{D/2} \left(|z| ^D + |z'| ^D \right)  + 1 \right) \nonumber
\\&\geq -  (\ell + m ) \left( \log |z| + \log |z'|\right) - D \log N - 2 \log \left( |z| ^D +1 \right) - 2 \log \left( |z| ^D +1 \right)
\end{align}
leads to
\begin{multline}\label{eq:MF main term b}
\int_{\R ^4}\left[  \half (|z| ^2 + |z'| ^2) + \frac{N}{2(N-1)} \Wal(z,z') dz dz' \right]\mut ^{(2)} (z,z') dz dz' 
\\ \geq - C (\ell + m) \log(\ell + m ) - C D \log N - C D\log D. 
\end{multline}

The term in the last line of ~\eqref{eq:low bound Fepal 2} is estimated as follows. Since $\mut$ is a probability, so is $\mut ^{(1)}$ and we have 
\begin{equation}\label{eq:MF error}
\int_{\R ^2} \mut ^{(1)} (z) \left((\ell + m) \logal (0) + \logal| g_2 (z,z)|\right) \geq C  ( \ell +m+D ) \left(\log \alpha + \frac12\right),
\end{equation}
where we used~\eqref{eq:log alpha 2}. 

Inserting now the representation~\eqref{eq:mut} into the main term proportional to $N$ in~\eqref{eq:low bound Fepal 2} we see that this part can be written as 
\beq N \int_{\rho \in \PP (\R ^2)} \tilde{\E} [\rho]\geq N\tilde E\eeq
where $\tilde{\E}$ is the modified mean-field functional 
\begin{multline*}
\tilde{\E} [\rho] := \int_{\R ^2} \left((1- 2N ^{-1}) |z| ^2 - \frac{2}{N-1} \log | g_1 (z)| + \eps U (z)\right) \rho(z) dz 
\\ + \frac{1}{2} \iint_{\R^4} \rho(z) \Wal(z,z') \rho(z') dz'
\end{multline*}
and $\tilde{E}$ its infimum over $\PP (\R ^2)$. We can then use a minimizer $\tilde{\rho}$ for this functional and obtain 
\begin{equation}\label{eq:MF main term 2}
 \tilde{E} = \MFfepal [\tilde{\rho}] - 2 N ^{-1} \intR |z| ^2 \tilde\rho \geq \MFeepal - 2 N ^{-1} \intR |z| ^2 \tilde\rho. 
\end{equation} 
The last term is again an error term, but up to this term and the error term give by~\eqref{eq:MF main term b} we have the desired lower bound. The final step is thus to estimate the term $2 N ^{-1} \intR |z| ^2 \tilde\rho$.\medskip

To this end we may on the one hand write  
\begin{multline*}
\tilde{\E} [\rho] = \int_{\R ^4} dz dz' \rho(z) \rho(z') \Big[ \frac{1}{2} \Wal(z,z') 
\\ + \frac{1}{2} \left((1-2N ^{-1}) (|z| ^2 +  |z'| ^2 ) - \frac{2}{N-1} \left(\log | g_1 (z)| +  \log | g_1 (z')|\right) + \eps U (z) + \eps U (z') \right) \Big]  
\end{multline*}
to obtain
\begin{multline*}
\tilde{\E} [\tilde \rho] \geq \left(\frac{1}{2} - 2 N ^{-1} \right) \intR |z| ^2 \tilde\rho + \inf_{\R ^4} \tilde{W} 
\\ \geq \frac{1}{2} \intR |z| ^2 \tilde\rho - C \left( 1 + D \log D + D \log N+ (\ell + m)\log (\ell + m) \right)
\end{multline*}
where $\tilde W$ is the two-body potential
\[
 \tilde{W} (z,z') = \frac{1}{4} \left( |z| ^2 + |z'| ^2 \right) -\frac{1}{2(N-1)} \left ( \log | g_1 (z)| + \log | g_1 (z')| \right) + \frac{1}{2} \Wal(z,z') 
\]
and a lower bound to its infimum is derived by elementary considerations similar to~\eqref{eq:low bound pot}. On the other hand, using as trial state the normalized characteristic function of a ball $B(0,1)$ centered at $0$ of radius 1,
\[
 \rho ^{\rm trial} = \frac{1}{|B(0,1)|} \one_{B(0,1)} 
\]
we easily have 
\[
 \tilde{E} \leq \tilde{\E} [\rho ^{\rm trial}] \leq C \left( 1+ D + \ell + m\right)
\]
from which 
\begin{equation}\label{eq:MF main term 3}
\intR |z| ^2 \tilde\rho \leq C \left( 1 + D \log \left(D\right) + (\ell + m)\log (\ell + m) \right) + D \log N 
\end{equation}
follows.

\end{proof}

\subsection{Conclusion: proof of Theorem~\ref{theo:incomp two}}\label{sec:incomp concl}

We recall that the $\mu_F$ we are interested in is equal to $\mu_{0,0}$ in the notation of Subsection~\ref{sec:prelim incomp}. Let $\ep >0$ be chosen small enough that the results of Sections~\ref{sec:prelim incomp} and~\ref{sec:MF lim} may be applied. We first write
\begin{multline}\label{eq:MF concl 1}
N \MFe_{0,0} +  N \eps \int_{\R ^2} U \mu_{0,0} ^{(1)}  + C  \geq \F_{0,0} [\mu_{0,0}] + N \eps \int_{\R ^2} U \mu_{0,0} ^{(1)}
\\ \geq  \F_{\eps,0} [\mu_{\eps,0}] \geq  \F_{\eps,\alpha} [\mu_{\eps,0}] \geq \F_{\eps,\alpha} [\mu_{\eps,\alpha}].
\end{multline}
The first inequality is proved using $(\rhoMF_{0,0}) ^{\otimes N}$ as a trial state for $\F_{0,0}$. The entropy term 
\[
 T \intRN (\rhoMF_{0,0}) ^{\otimes N} \log \left( (\rhoMF_{0,0}) ^{\otimes N} \right) =  N T \intR \rhoMF_{0,0} \log \rhoMF_{0,0}
\]
is bounded above using~\eqref{eq:bound MF} and recalling that the temperature is of order $N ^{-1}$. The other inequalities in~\eqref{eq:MF concl 1} use either the variational principle or $-\log |z|\geq -\log_\alpha |z|$.

Next, we use first Proposition~\ref{pro:estim free ener} and then Lemma~\ref{lem:alpha to 0} to obtain
\begin{align*}
\F_{\eps,\alpha} [\mu_{\eps,\alpha}]  & \geq N \MFeepal - C \left(\ell +m+D \right)\left(| \log \alpha | + 1 \right)\\
&-C \left( D\log D + D \log N + (\ell + m) \log (\ell + m) \right)\\
&\geq N \MFe_{\eps,0}  - C \left(\ell +m +D \right)\left(| \log \alpha | + 1 \right) - 
 CN(1+\eps\,{\rm sup}\, |\Delta U|)\,\alpha^{\min(2,2/D)}
\\
&-C \left( D\log D + D \log N + (\ell + m) \log (\ell + m) \right).
\end{align*}
Then 
\begin{equation*}
\MFe_{\eps,0} = \MFf_{0,0}[\rhoMF_{\eps,0}]  + \eps \int_{\R ^2} U  \rhoMF_{\eps,0} \geq \MFe_{0,0} + \eps  \int_{\R ^2} U  \rhoMF_{\eps,0}
\end{equation*}
by the variational principle applied to $\MFf_{0,0}$. Summing up we have (recall that $\eps >0$)
\begin{align}\label{eq:final density estimate}
\int_{\R ^2} U \mu_{0,0} ^{(1)} &\geq  \int_{\R ^2} U  \rhoMF_{\eps,0} \nonumber
\\ &- C (N\eps) ^{-1} \left( 1 +  \left(\ell + m+D \right)\left(| \log \alpha | + 1 \right) +  N  (1+\eps\,{\rm sup}\, |\Delta U|)\,\alpha^{\min(2,2/D)}
 \right) \nonumber
\\ & -C (N\eps) ^{-1} \left( D\log D + D \log N + (\ell + m) \log (\ell + m)\right) \nonumber
\\ &\geq  \int_{\R ^2} U  \rhoMF_{\eps,0} - C (N\eps) ^{-1} \left( 1  + (\ell + m + D) \log N \right)
\\ & -C (N\eps) ^{-1} \left(D\log D + D \log N +  (1+\eps\,{\rm sup}\, |\Delta U|)\right) 
\end{align}
where we have chosen $\alpha = N^{-D/2}$ if $D\geq 1$ and $\alpha = N ^{-1/2}$ if $D=0$ to obtain the last inequality. Recalling that $\mu_{0,0} = \muF$, the lower bound~\eqref{eq:incomp two} is proved, with $\rhoF := \rhoMF_{\eps,0}$. This density satisfies~\eqref{eq:bound MF thm} by Lemma~\ref{lem:MF func}. 

\hfill \qed

\medskip

\begin{remark}[\textbf{The opposite inequality}]\label{rem:converse incomp}\mbox{}\\
Following the very same steps as above but using a perturbation potential $-\eps U$ instead of $\eps U$ in the Hamiltonian one obtains
\begin{equation}\label{eq:incomp two sup}
\int_{\R ^2} U \muFone \leq  \int_{\R ^2} U  \hat{\rho}_F + C\left(N\eps \right) ^{-1} \error (m,D,\eps U).
\end{equation} 
with $\hat{\rho}_F := \rhoMF_{-\eps,0}$. This density also satisfies~\eqref{eq:bound MF thm} and the above estimate is thus a kind of  converse to~\eqref{eq:incomp two}. Note that the reference density is different however. We shall not use this remark anywhere in the paper.
\hfill\qed

\end{remark}

\section{Conclusion of the proofs}\label{sec:proof concl}

Here we conclude the proofs of Theorem~\ref{thm:incomp main} and Corollary~\ref{cor:radial}, in Subsections~\ref{sec:concl main} and~\ref{sec:cor proof} respectively.

\subsection{Response to external potentials: proof of Theorem~\ref{thm:incomp main}}\label{sec:concl main}

We now  bound from below the energy $\E_N [\Psi_F]$ when $F \in \VD$ as defined in~\eqref{eq:var set 2 D}. Without loss, we write $F$ as in~\eqref{eq:F two correl}, with 
\begin{equation}\label{eq:assum step 2}
m+\deg(f_2)\leq D \mbox{ and } \deg(f_1) \leq D N. 
\end{equation}
and assume that~\eqref{eq:bounds f1 f2} holds. We will apply the analysis of Section~\ref{sec:incompressibility}. Let us pick a large constant $B$ (to be tuned later on) and define the truncated potential
\begin{equation}\label{eq:truncated potential}
V_B (\xbf):= \min\{V(\xbf),B\}. 
\end{equation}
Thanks to~\eqref{eq:increase V}, this potential is constant outside of some ball centered at the origin and satisfies the assumption of Theorem~\ref{theo:incomp two}. We may thus apply this result with $U=V_B$ and the correlation factor $F$ at hand. In view of~\eqref{eq:defi Gibbs} we have 
\[
\mu_F ^{(1)} (z) = (N-1) \rhoPF \left( \sqrt{N-1}\: z \right) 
\]
and the theorem implies that there exists a $\rhoF$ of unit $L ^1$ norm satisfying
\[
0\leq \rhoF \leq \frac{1}{\pi(\ell + m)} + \frac{\eps \: \sup |\Delta V_B| }{\pi (\ell + m)}
\]
such that 
\begin{align}\label{eq:appli incomp two}
\E_N [\Psi_F] &= (N-1) \intR V_B (\xbf) \rhoPF \left( \sqrt{N-1}\: \xbf \right)d\xbf  \nonumber \\
&\geq  \int_{\R ^2} V_B  \rhoF - C (N\eps) ^{-1} \error (m,D,\eps V_B).
\end{align}
Here, by assumption, $m$ and $D$ are fixed when $N\to \infty$. Passing then to the limit $N\to \infty$ at fixed $\eps$ we obtain
\[
\liminf_{N\to \infty}  \E_N [\Psi_F] \geq \inf\left\{ \intR V_B \rho,\: 0\leq \rho \leq \frac{1}{\pi(\ell + m)} + \frac{\eps \: \sup |\Delta V_B| }{\pi (\ell +m)} \right\}.
\]
We may then pass to the limit $\eps \to 0$: 
\[
\liminf_{N\to \infty}  \E_N [\Psi_F] \geq \inf\left\{ \intR V_B \rho,\: 0\leq \rho \leq \frac{1}{\pi(\ell + m)} \right\}
\]
and finally to the limit $B\to \infty$, which yields
\[
\liminf_{N\to \infty}  \E_N [\Psi_F] \geq E_V \left(\ell +  m \right) \geq E_V (\ell).
\]
The necessary continuity of the bath-tub energy~\eqref{eq:bath tub} as a function of the upper bound on the admissible trial states and the cut-off of the potential are easily deduced from the explicit formulae of, e.g.,~\cite[Theorem 1.14]{LL}. In fact, if $B$ is large enough the bath-tub energy in the truncated potential $V_B$ is equal to the bath-tub energy in the potential $V$.

\hfill \qed

\medskip

\subsection{Optimality in radial potentials: proof of Corollary~\ref{cor:radial}}\label{sec:cor proof} 

Given Theorem~\ref{thm:incomp main}, the only thing left to do is the proof of~\eqref{eq:trial increas} and~\eqref{eq:trial mexican 2}. We thus consider the special trial functions~\eqref{eq:trial mexican 1}, that are built using correlation factors of the form 
\begin{equation}\label{eq:F GV Lau}
F(z_1,\ldots,z_N) = \prod_{j=1} ^N z_j ^m, \quad m \in \N.  
\end{equation}
The corresponding classical Hamiltonians~\eqref{eq:class hamil} that we associate to them in order to analyze the one-body densities of the states $\Psi_F$ have purely Coulomb two-body interactions ($g_2 \equiv 1$ in this case). One can take advantage of this fact to analyze the classical mean-field limit with a different method than that we used in Section~\ref{sec:incompressibility}. This was done in~\cite{RSY2} and the method leads to somewhat stronger estimates. 

We continue to use the notation of Section~\ref{sec:incompressibility} to quote some results from~\cite{RSY2}. Thus $\mu_F ^{(1)}$ is the rescaled one-body density of the state $\Psi_F$. Ultimately it will be sufficient to take $m\sim CN$ in the limit $N\to \infty$, so we may\footnote{In~\cite{RSY2} we only considered the case $\ell = 2$, but the whole analysis adapts to any fixed $\ell$. Lengths were scaled by a factor $\sqrt{N}$ instead of $\sqrt{N-1}$, so the formulas we quote have to be slightly modified.} invoke~\cite[Theorem 3.1, Item 1]{RSY2}: for any regular enough $U:\R ^2 \mapsto \R$
\begin{equation}\label{eq:density GV}
\left\vert \intR U \left(\muFone - \rhoMF\right)  \right\vert \leq  C N^{-1/2} \log N \Vert \nabla U \Vert_{L ^2 (\R ^2)} + C N ^{-1/2}\Vert \nabla U \Vert_{L ^{\infty} (\R ^2)}
\end{equation}
where ``regular enough'' only means that we require the norms appearing on the right-hand side of the above equation to be finite. Here $\rhoMF$, which was denoted $\rhoMFel$ in~\cite{RSY2}, is the minimizer of the functional~\eqref{eq:MF func} corresponding to the choice~\eqref{eq:F GV Lau}. This means that we take $m = 0$, $g_2 \equiv 1$ and 
$$ -\frac{1}{N-1} \log |g_1 (z) |= - \frac{m}{N-1} \log |z|. $$
In this case, $\rhoMF$ can be explicitly computed, see~\cite[Proposition 3.1]{RSY2}:
\begin{align}\label{eq:MF GV Lau}
\rhoMF &= \frac{1}{\ell \pi} \one_{B(0,\sqrt{\ell})} \mbox{ if } m=0  \\ 
\rhoMF &= \frac{1}{\ell \pi} \one_{\AN} \mbox{ if } m>0 \label{eq:exp rhoMFmh}
\end{align}
where $\AN$ is the annulus of inner radius $\Rminus = \sqrt{m/(N-1)}$ and outer radius $\Rplus = \sqrt{(\ell + m)/(N-1)}$ centered at the origin. 

Of course we cannot apply~\eqref{eq:density GV} directly with $U=V$, since the norms involved in the estimate are infinite for the latter potential. We thus first split $V$ in two parts
\begin{equation}\label{eq:split V}
V(\xbf) = \chiin (\xbf) V(\xbf) + \chiout (\xbf) V(\xbf)
\end{equation}
using a smooth partition of unity $\chiin + \chiout \equiv 1$, where $\chiin = 1$ in $B(0,R)$ and $\chiin = 0$ in $B(0,2R) ^c$  for some $R$ to be chosen later on. We will use~\eqref{eq:density GV} to deal with the $\chiin V$ part, and show the contribution of the $\chiout V$ part to be negligible using 
\begin{multline}\label{eq:1 particle decay}
\muFone (z) \leq C_1 \exp\left( -C_2 N \left(\left( |z|- \sqrt{\frac{m}{N-1}} \right)^2 -  \log N \right)\right) 
\\ \mbox{when } \left||z|-\sqrt{\frac{m}{N-1}} \right| \geq C_3. 
\end{multline}
which is~\cite[Equation (3.16)]{RSY2} and we recall that $m \sim CN$ in our case. This estimate implies that for any power $\alpha >0$ and $N$ large enough
\begin{equation}\label{eq:1 particle decay simple}
\muFone (z) \leq C_1 \exp\left( -C_2 N |z| ^2  \right) 
\\ \mbox{when } |z| \geq N ^{\alpha}, 
\end{equation}
where the value of the constants $C_1,C_2$ have changed. Choosing $R= N ^{\alpha}$ for some small (but fixed) power $\alpha>0$ we then clearly have 
\begin{equation}\label{eq:contribution out}
(N-1) \intR \chiout (\xbf) V(\xbf) \rhoPF \left( \sqrt{N-1} \: \xbf\right) d\xbf = \intR \chiout (\xbf) V(\xbf) \muFone  \left(\xbf\right) d\xbf \to 0
\end{equation}
when  $N \to \infty$. This is the one place where we use the assumption that $V$ grows at most polynomially at infinity (which, in view of~\eqref{eq:1 particle decay simple} could be relaxed a bit).

Next we can use~\eqref{eq:density GV} to show that 
\begin{align}\label{eq:contribution in}
(N-1) \intR \chiin (\xbf) V(\xbf) \rhoPF \left( \sqrt{N-1} \: \xbf\right) d\xbf &= \intR \chiin (\xbf) V(\xbf) \muFone  \left(\xbf\right) d\xbf \nonumber
\\ &\sim \intR V(\xbf) \rhoMF (\xbf) d\xbf \mbox{ when } N \to \infty.
\end{align} 
We have already estimated very similar terms in~\cite[Section 4]{RSY2} and will not reproduce all computations here. Simply observe that $\chiin$ can clearly be chosen with $|\nabla \chiin| \leq C N ^{\al}$. Also, since we are free to choose the power $\alpha$ as small as desired, using the assumption that $V$ grows at most poynomially at infinity, we may guarantee that the norms of $U = \chiin V$ that appear in the right-hand side of~\eqref{eq:density GV} grow at most as $N ^{c\alpha}$ for some constant $c>0$. This allows to control the error term as in~\cite[Section 4]{RSY2}, and proves~\eqref{eq:contribution in}.

Gathering~\eqref{eq:split V},~\eqref{eq:contribution out} and~\eqref{eq:contribution in} we have now proved that 
\begin{equation}\label{eq:up bound final}
\E_N [\Psi_F] \sim \intR V(\xbf) \rhoMF (\xbf) d\xbf \mbox{ when } N \to \infty 
\end{equation}
where $F$ is as in~\eqref{eq:F GV Lau} and we have assumed $m\sim C N$. We now discuss this final result as a function of 
\begin{equation}\label{eq:m bar}
\mb = \lim_{N\to \infty} \frac{m}{N}. 
\end{equation}
We have, in $L ^{\infty}$ norm,
\begin{align}\label{eq:MF GV Lau lim}
\rhoMF &\to \frac{1}{\ell \pi} \one_{B(0,\sqrt{\ell})} \mbox{ if } \mb =0  \\ 
\rhoMF &\to \frac{1}{\ell \pi} \one_{  \sqrt{\mb} \leq |\xbf| \leq \sqrt{\ell + \mb} } \mbox{ if } \mb > 0 \label{eq:exp rhoMFmh lim},
\end{align}
which allows to conclude the proof. Indeed, it is well-known~\cite[Theorem 1.14]{LL} that the infimum in the bath-tub energy~\eqref{eq:bath tub} is attained for a density saturating the constraint $\rho \leq (\pi \ell ) ^{-1}$. The minimizer is explicit as a function of the potential $V$, and in the cases we consider here it is easy to see that it is exactly equal to~\eqref{eq:MF GV Lau lim}, provided a proper choice of $\mb$ is made.

In case (1) of Corollary~\ref{cor:radial} we take the pure Laughlin state, i.e. $F\equiv 1$ and so $\mb = 0$, and we obtain 
\[
\E_N [\Psi_F] \to \frac{1}{\ell \pi} \int_{B(0,\sqrt{\ell})} V(\xbf) d\xbf \mbox{ when } N \to \infty 
\]
and the latter quantity is of course equal to $E_V (\ell)$, the bath-tub energy defined in~\eqref{eq:bath tub}. Indeed, if $V$ is radial increasing, the minimizer of the bath-tub energy is simply~\eqref{eq:MF GV Lau lim}: the density has to saturate the bound $\rho \leq (\pi \ell ) ^{-1}$ on its support, and it is clear that the optimal choice is to take this support to be a disc centered on the minimum of $V$, i.e at the origin. This proves~\eqref{eq:trial increas}.

In case (2), a similar reasoning yields that the minimizer of $E_V (\ell)$ is given by
\[
\frac{1}{\ell \pi} \one_{  A \leq r\leq B } 
\]
for some $A,B>0$ tuned so that the above function is normalized in $L ^1$, i.e. $B ^2 - A ^2 = \ell$. Choosing $\mb$ so that 
$$A = \sqrt{\mb} \mbox{ and } B=\sqrt{\ell + \mb},$$
that is, taking $m$ to be the integer part of $N A ^2 $ in~\eqref{eq:trial mexican 1} we deduce that also in this case 
\[
\E_N [\Psi_F] \to E_V(\ell) \mbox{ when } N \to \infty, 
\]
which is~\eqref{eq:trial mexican 2}.

If $D$ is large enough, the trial states we have just built indeed all belong to $\VD$, so we deduce from~\eqref{eq:trial increas}-\eqref{eq:trial mexican 2} that 
$$ \ED \leq \E_N [\Psi_F] \to E_V(\ell)$$
which combines with~\eqref{eq:main incomp} to complete the proof of~\eqref{eq:incomp opt}.

\hfill \qed

\section{Extensions of the main results}\label{sec:extensions}

Energy lower bounds of the type~\eqref{eq:main incomp} are a manifestation of the incompressibility of the states involved and ideally one would like to derive them for all fully correlated states~\eqref{eq:Ker}, not only the special states considered in Theorem~\ref{thm:incomp main}. This is a quite ambitious goal and genuinely new ideas will be needed to achieve it completely. However, substantial generalizations of Theorem~\ref{thm:incomp main} can be achieved by our methods as we now discuss.

The main generalization one could handle with our methods corresponds to allow ``$n$-body correlation factors'' of the form 
\begin{equation}\label{eq:choice correlations n}
F (z_1,\ldots,z_N) = \prod_{j = 1} ^N f_1 (z_j)\prod_{( i,j ) \in \{1,\ldots,N \} } f_2 (z_i,z_j)\ldots \prod_{( i_1,\ldots,i_n ) \in \{1,\ldots,N \} } f_n (z_{i_1},\ldots,z_{i_n})
\end{equation}
for some finite fixed $n$ and symmetric holomorphic functions $f_1,\ldots,f_n \in \Barg ^1, \ldots, \Barg ^n$ that can depend on $N$, with no a priori bounds on their degree. Let us sketch briefly these possible improvements:

\medskip

\noindent \textbf{Removing a priori bounds on the degree.} In~\eqref{eq:var set 2 D} we have restricted our attention to polynomials $f_1$  and $f_2$ satisfying a priori bounds on their degrees. First note that the estimates of Theorem~\ref{theo:incomp two} are actually explicit as a function of this degree so the theorem is still valid if the assumption is relaxed, with worse remainder terms however.

In fact we claim that if either bound is violated for a sequence 
$$F (z_1,\ldots,z_N) = \prod_{j = 1} ^N f_1 (z_j)\prod_{( i,j ) \in \{1,\ldots,N \} } f_2 (z_i,z_j)$$
then the rescaled density defined by~\eqref{eq:defi Gibbs} satisfies
\begin{equation}\label{eq:vanishing}
\mu_F ^{(1)} \wto 0 
\end{equation}
weakly as measures, which clearly implies 
$$ \E_N [\Psi_F] \to +\infty \mbox{ when } N \to \infty$$
because of~\eqref{eq:increase V} and justifies the a priori reduction we made. 

More precisely in this case one has
\begin{equation}\label{eq:vanishing 2}
 \mu_F ^{(1)} (B (0,R)) \to 0 \mbox{ when } N \to \infty 
\end{equation}
for any fixed radius $R$. The computations leading to~\eqref{eq:vanishing 2} are a bit tedious, especially for general polynomials $f_1$ and $f_2$, but here is the main idea: If either $g_1$ or $g_2$ has a large degree in the classical Hamiltonian~\eqref{eq:class hamil}, then it is clear that the two-body potential 
\begin{multline}\label{eq:potential generalize}
 \bar{W} (z,z') := \frac{|z| ^2}{2} + \frac{|z'| ^2}{2} - \frac{2}{N-1} \log |g_1 (z)| - \frac{2}{N-1} \log |g_1 (z)| 
 \\ - 2 (\ell + m ) \log |z-z'| - 2 \log |g_2 (z,z')| 
\end{multline}
takes its minimum far from the origin, in fact infinitely far in the limit $N\to \infty$. Simple model cases are e.g. $f_1 \equiv 1$, $f_2 \equiv 1$ and $m\to \infty$ when $N\to \infty$, or\footnote{This latter case has been considered in~\cite{RSY2}.} $f_1 (z) = z ^\gamma$, $f_2 \equiv 1$ and $m=0$ with $\gamma \gg N$.
Since the free-energy of $\mu_F$ in the plasma analogy is 
\[
 N \iint_{\R ^2 \times \R ^2} \bar{W} (z,z') \muF ^{(2)} (z,z') dzdz' + \frac{1}{N-1} \intRN \muF \log \muF
\]
it is clear that it will be favorable for $\muF$ to have its mass concentrated far from the origin. One can easily construct factorized trial states having this behavior and compare their free-energy to that of the minimizer $\muF$, taking into account that $\bar{W}$ is much larger than its infimum in any ball $B(0,R)$ with $R$ fixed. Rather simple arguments then allow to deduce~\eqref{eq:vanishing 2} in good cases, but the computations allowing to control~\eqref{eq:potential generalize} for generic polynomials $f_1,f_2$ are rather tedious. Arguments of this sort could also be adapted to rule out more general holomorphic functions than polynomials.


\medskip

\noindent\textbf{Higher order correlation factors.} Recall the definition~\eqref{eq:BargN} of the $N$-body bosonic Bargmann space and define the set  
\begin{multline}\label{eq:var set n}
\V _n = \Big\{ F \in \Barg ^N, \mbox{ there exist } (f_1,\ldots,f_n) \in \Barg \times \ldots \times \Barg ^n,
\\  F (z_1,\ldots,z_N) = \prod_{j = 1} ^N f_1 (z_j)\prod_{( i,j ) \in \{1,\ldots,N \} } f_2 (z_i,z_j)\ldots \prod_{( i_1,\ldots,i_n ) \in \{1,\ldots,N \} } f_n (z_{i_1},\ldots,z_{i_n}) \Big\}
\end{multline}
where $( i_1,\ldots,i_n )$ is understood as an $n$-tuple with no repetition of any of the indices $i_k,k\in \{1,\ldots,n \}$. We have 
\begin{itemize}
\item $\V_n \subset \V_{n+1}$ : given $F_n\in \V_n$ associated to $(f_1,\ldots,f_n)\in \Barg \times \ldots \times \Barg ^n$ we can see it as an element of $\V_{n+1}$ associated with $(f_1,\ldots,f_n,1) \in \Barg \times \ldots \times \Barg ^n \times \Barg^{n+1}$.
\item $\V_N = \BargN$ : for $F_N \in \BargN$ one may simply choose $(1,1,\ldots,F_N)\in \Barg \times \ldots \times \Barg ^N$.
\end{itemize}
We can consider the minimization of~\eqref{eq:start energy} amongst states with the correlation factors~\eqref{eq:var set n} instead of the simpler~\eqref{eq:var set 2 D}. This way we obtain the family of energies
\begin{equation}\label{eq:energy Nn extend}
 E_n (N):= \inf \left\{ \E_N [\Psi_F], \; \Psi_F \mbox{ of the form } (\ref{eq:PsiF}) \mbox{ with } F \in \V _n \right\}.
\end{equation}
and since $\V_n \subset \V_{n+1}$ we have of course  
\begin{equation}\label{eq:cascade n}
E (N) = E_N (N) \leq \ldots \leq E_{n+1} (N) \leq E_n (N) \leq \ldots \leq E_N (1). 
\end{equation}
One may expect that equality holds, at least asymptotically for large $N$, in a large number of these equalities. A first step would be to generalize Theorem~\ref{thm:incomp main} to a lower bound on $E_n (N)$ for $n$ as large as possible.

In case $F\in \V_n \setminus \V_{n-1}$ for some fixed finite $n \geq 3$ one should in fact also expect  
\begin{equation}\label{eq:vanishing bis}
\mu_F ^{(1)} \wto 0, 
\end{equation}
with the same consequences as before. Indeed, the plasma analogy also applies to functions built on $\V_n$, leading now to a classical energy including $n$-body terms, but now because of the symmetry there are at least ${N \choose n}$ interaction terms. Because of the scaling in~\eqref{eq:defi Gibbs}, these come multiplied by a prefactor $\propto N ^{-1}$, so that the total interaction strength will be of order roughly $N ^{-1}{N \choose n} \gg 1$ if $n>2$. With such a huge strength, it is intuitive that~\eqref{eq:vanishing bis} should occur. In fact, to see a non trivial behavior of $|\Psi_F| ^2$, one should rather rescale it on a much larger length-scale. To make this more precise one can inspect the resulting $n$-body potential $\bar{W}_n$ replacing~\eqref{eq:potential generalize} and see that its minimum would occur again very far from the origin. If $n\geq 3$, this will happen even if the $n$-body correlation factors have bounded degree. Since the classical free-energy of the classical plasma will now 
have the form 
\[
 N \int_{\R ^{2n}} \bar{W}_n (z_1,\ldots,z_n) \muF ^{(n)} (z_1,\ldots,z_n) dz_1,\ldots,dz_n + \frac{1}{N-1} \intRN \muF \log \muF
\]
one can argue as before that~\eqref{eq:vanishing bis} occurs. 

Of course these (formal) arguments break down if only the correlation factors $f_n$, $n\geq N-1$ are non trivial. In this case the total interaction strength is again of order $N ^{-1}{N \choose n}\lesssim 1$, so even if these ideas could be made rigorous, they would not allow to deal with the full set $\Ker$. One could still investigate for example the case of $n$ fixed when $N\to \infty$ and deduce that in this case
$$ \liminf_{N\to \infty} E_n (N) \geq E_V (\ell).$$
Combining with the arguments of Corollary~\ref{cor:radial}, one would then obtain 
$$ \lim_{N\to \infty} E_n (N) = \lim_{N\to \infty} E_m (N)$$ 
for any fixed $n,m\in \N$ in the case of radial increasing or radial mexican-hat potentials.

Concerning the analysis of the states~\eqref{eq:choice correlations n} we finally note that one could also extract some useful information from the mean-field approximation procedure of Section~\ref{sec:incompressibility}. It would be useful to scale space variables differently but this is a detail. One would need to work with the formulas~\eqref{eq:DiacFreed 2} for higher-order marginals in the Diaconis-Freedman theorem. In view of~\eqref{eq:DiacFreed}, it should be possible to obtain quantitative information as long as~$n\ll \sqrt{N}$.

\appendix

\section{The Diaconis-Freedman theorem}\label{sec:DiacFreed}

We make use of the Diaconis-Freedman theorem~\cite{DF}, which may be seen as a quantitative version of the Hewitt-Savage (or classical de Finetti) theorem~\cite{HS}, whose importance for classical mean-field problems has been recognized for some time now~\cite{CLMP,Kie1,KS,MS,Gol}. The Hewitt-Savage theorem is most often seen as an existence result in the literature, but it in fact follows from the constructive approach in~\cite{DF}. Since we make use of this fact, it is worth recalling the main result of~\cite{DF} and sketching the proof. We shall denote $\Vert \: . \: \Vert_{\rm TV}$ the total variation norm.

\begin{theorem}[\textbf{Diaconis-Freedman}]\label{thm:DiacFreed}\mbox{}\\
Let $S$ be a measurable space and $\mu \in \PP_s (S ^N)$ be a probability measure on $S^N$ invariant under permutation of its arguments. There exists $P_{\mu} \in \PP(\PP (S))$ a probability measure such that, denoting 
\begin{equation}\label{eq:Pnu}
\mut := \int_{\rho \in \PP (S)} \rho ^{\otimes N} dP_{\mu} (\rho)
\end{equation}
we have
\begin{equation}\label{eq:DiacFreed}
\left\Vert \mu ^{(n)} - \mut ^{(n)} \right\Vert _{\rm TV} \leq \frac{n(n-1)}{N}.
\end{equation}
In addition, the marginals of $\mut$ are given by :
\begin{equation}\label{eq:DiacFreed 2}
\mut  ^{(n)} (x_1,\ldots,x_n) =  \frac{1}{N ^n} \sum_{j=1} ^n \frac{N\,!(n-j)\,!}{(N-j)\,!\,n\,!}\sum_{\sigma \in \Sigma_n} \mu ^{(j)} (x_{\sigma(1)},\ldots,x_{\sigma(j)}) \: \delta_{x_{\sigma(j)} = x_{\sigma(j+1)} = \ldots = x_{\sigma(n)}}
\end{equation}
with $\Sigma_n$ the group of permutations of $n$ elements.
\end{theorem}

\begin{proof}
The proof may be found in~\cite{DF}, 
we provide a sketch for the convenience of the reader. We take $S = \R ^2$, which is the case of interest for us here and abuse notation by writing $\mu(Z) dZ$ instead of $d\mu(Z)$ for integrals in $Z\in \R ^{2N}$. Note that the symmetry of $\mu$ implies 
\begin{equation}\label{eq:repre P}
\mu (X) = \intRN \mu (Z) \sum_{\sigma \in \Sigma_N} (N!) ^{-1} \delta_{Z_{\sigma} = X } dZ.
\end{equation}
The main idea is to define 
\begin{equation}\label{eq:defi Pnu}
\mut (X) = \intRN  \mu(Z) \sum_{\gamma \in \Gamma_N} N ^{-N} \delta_{Z_{\gamma} = X} dZ,
\end{equation}
with $\Gamma_N$ the set of all maps (not necessarily one-to-one) from $\left\{1,\ldots,N \right\}$ to itself. Noticing that 
\begin{equation}\label{eq:factor urn}
\sum_{\gamma \in \Gamma_N} N ^{-N} \delta_{Z_{\gamma} = X}  = \left( N ^{-1} \sum_{j=1} ^N \delta_{z_j=x} \right) ^{\otimes N}, 
\end{equation}
one may put~\eqref{eq:defi Pnu} in the form~\eqref{eq:Pnu} by taking 
\begin{equation}\label{e:defi nu}
P_{\mu} (\rho) = \intRN \delta_{\rho = \bar{\rho}_Z} \mu(Z) dZ, \quad  \bar{\rho} _Z (x) := \sum_{i=1} ^N N ^{-1} \delta_{z_j=x}.
\end{equation}
Computing the difference between $\mu ^{(n)}$ and $\mut ^{(n)}$ we have of course 
\[
\mu ^{(n)} - \mut ^{(n)} =  \intRN  \left( \left(\sum_{\sigma \in \Sigma_N} (N!) ^{-1} \delta_{Z_{\sigma} = X}\right) ^{(n)} -  \left(\sum_{\gamma \in \Gamma_N} N ^{-N} \delta_{Z_{\gamma} = X}\right) ^{(n)}\right) \mu(Z) dZ
\]
but $\sum_{\sigma \in \Sigma_N} (N!) ^{-1} \delta_{Z_{\sigma} = X}$ is the probability law of drawing $N$ balls at random from an urn \footnote{Where the balls are labeled $x_1,\ldots,x_N$.}, \textit{without} replacement whereas $\sum_{\gamma \in \Gamma_N} N ^{-N} \delta_{Z_{\gamma} = X}$ is the probability law of drawing $N$ balls at random from an urn, \textit{with} replacement. It is thus intuitively clear that the difference 
\[
\left(\sum_{\sigma \in \Sigma_N} (N!) ^{-1} \delta_{Z_{\sigma} = X}\right) ^{(n)} -  \left(\sum_{\gamma \in \Gamma_N} N ^{-N} \delta_{Z_{\gamma} = X}\right) ^{(n)} 
\]
between their reduced densities is small when $n\ll N$. The meaning of `small' in this sentence is not difficult to quantify as a function of $n$ and $N$, see~\cite{Fre} where the total variation bound $\frac{n(n-1)}{N}$ is obtained, which leads to~\eqref{eq:DiacFreed}.

The only fact that is not explicitly mentioned in~\cite{DF} is~\eqref{eq:DiacFreed 2}, but this is an easy consequence of~\eqref{eq:factor urn}. 
Using the symmetry of $P$ we have
\begin{align}\label{eq:marginals DF}
\mut ^{(1)} (x) &= N ^{-1} \sum_{j=1} ^N \intRN \mu (Z) \delta_{z_j=x} dZ = \mu ^{(1)} (x)\nonumber\\
\mut ^{(2)} (x_1,x_2) &= N ^{-2} \intRN \mu (Z) \left( \sum_{j=1} ^N \delta_{z_j = x_1} \right) \left( \sum_{j=1} ^N \delta_{z_j = x_2} \right) dZ \nonumber \\
&= N ^{-2} \sum_{ 1 \leq i \neq j \leq N} \intRN \mu (Z) \delta_{ z_i = x_1} \delta_{z_j = x_2}  dZ + N ^{-2} \sum_{i=1} ^N \intRN \mu (Z) \delta_{z_i = x_1} \delta_{z_i = x_2} dZ \nonumber \\
&= \frac{N-1}{N} \mu ^{(2)} (x_1,x_2) + \frac1N \mu ^{(1)} (x_1) \delta_{x_1 = x_2}.
\end{align}
%
%
The computation of the higher order marginals follows along the same lines and leads to~\eqref{eq:DiacFreed 2}. An estimate of the form~\eqref{eq:DiacFreed} can also be seen as following from this computation as noted by Lions~\cite{Lio}.

\end{proof}

\end{document}